\documentclass[aps,pra,reprint,groupedaddress]{revtex4-1}

\usepackage{textcomp}
\usepackage{amssymb}
\usepackage{graphicx}
\usepackage{amsmath}
\usepackage{amssymb}
\usepackage{yfonts}
\usepackage{enumitem}
\usepackage[font=small,format=plain,up,textfont=small,up,justification=raggedright,singlelinecheck=false]{caption}
\usepackage{float}
\usepackage{url}
\usepackage{color}

\begin{document}

%Theorem etc environment defintions
\newtheorem{theorem}{Theorem}[section]
\newtheorem{lemma}[theorem]{Lemma}
\newtheorem{proposition}[theorem]{Proposition}
\newtheorem{corollary}[theorem]{Corollary}
\newtheorem{definition}[theorem]{Definition}

\newenvironment{proof}[1][Proof]{\begin{trivlist}
\item[\hskip \labelsep {\bfseries #1}]}{\end{trivlist}}
\newenvironment{example}[1][Example]{\begin{trivlist}
\item[\hskip \labelsep {\bfseries #1}]}{\end{trivlist}}
\newenvironment{remark}[1][Remark]{\begin{trivlist}
\item[\hskip \labelsep {\bfseries #1}]}{\end{trivlist}}

\newcommand{\qed}{\nobreak \ifvmode \relax \else
      \ifdim\lastskip<1.5em \hskip-\lastskip
      \hskip1.5em plus0em minus0.5em \fi \nobreak
      \vrule height0.75em width0.5em depth0.25em\fi}

%\preprint{}

\title{Classical simulation of measurement-based quantum computation on higher-genus surface-code states}


\author{Leonard Goff}
\author{Robert Raussendorf}
%\thanks{}
\affiliation{Department of Physics and Astronomy, University of British Columbia, Vancouver, British Columbia V6T 1Z1, Canada}


\date{\today}

\begin{abstract}
We consider the efficiency of classically simulating measurement-based quantum computation on surface-code states. We devise a method for calculating the elements of the probability distribution for the classical output of the quantum computation. The operational cost of this method is polynomial in the size of the surface-code state, but in the worst case scales as $2^{2g}$ in the genus $g$ of the surface embedding the code. However, there are states in the code space for which the simulation becomes efficient. In general, the simulation cost is exponential in the entanglement contained in a certain effective state, capturing the encoded state, the encoding and the local post-measurement states. The same efficiencies hold, with additional assumptions on the temporal order of measurements and on the tessellations of the code surfaces, for the harder task of sampling from the distribution of the computational output.
\end{abstract}

%\keywords{}

\maketitle


\section{Introduction}

A major open problem in quantum computation is to
determine the physical properties of quantum systems
that account for the quantum speedup
over classical computation. This would aid in the development
of useful quantum computational systems, and
constitute a significant leap forward in our understanding
of quantum physics.

One approach to studying this problem is to find instances of quantum computational processes that can be simulated efficiently on a classical computer, and identify which quantum mechanical properties they lack. There are three known examples in this category; namely quantum circuits composed of Clifford gates \cite{gottesman}, matchgate circuits \cite{matchgates}, \cite{valiant} (which can be mapped to non-interacting fermions \cite{terhald}), and quantum evolutions in which the entanglement--as quantified by an appropriate monotone--always remains small \cite{vidalslight}, \cite{JozLi}. 

Specifically, it was shown in \cite{vidalslight} that any circuit model quantum computation can be classically simulated with a number of steps that grows polynomially in the number of qubits, but exponentially in an entanglement measure $\chi$. Therein, $\chi$ is the log of the maximum value of the Schmidt rank across any bipartition of the set of qubits, at any point of the computation. This result has counterparts in measurement-based quantum computation (MBQC) \cite{Shi}, \cite{vdnetal}. However, such results relating the amount of entanglement present in a quantum system to the hardness of its classical simulation need to be taken with a grain of salt: they do not hold for all entanglement measures. Specifically, they do not hold for sufficiently continuous entanglement measures \cite{vdN12}. Also note that quantum states can be too entangled to be useful for MBQC \cite{tooentangled}, \cite{Winter}.

In this paper we describe a classical simulation method for quantum systems that combines the fermionic or matchgate method with that for slightly entangled quantum systems. To this end, we consider the classical simulation of MBQC where the initial resource state is a state in the code space of the surface code. The originally intended application for surface codes is fault-tolerant quantum computation in two-dimensional local architectures with constrained interaction range \cite{kitaev}, \cite{kitaevqm}. Regarding the potential use of surface-code states as resources in MBQC, it was previously shown that for such states with a planar topology the resulting quantum computation can be efficiently classically simulated \cite{vdnspinmodels},\cite{rausstoric}. 

Here, we extend this investigation to surface-codes embedded in surfaces of higher genus. This problem is related to, but not the same as matchgate contraction \cite{bravyitensor} and computing the Ising model partition function \cite{nonplanarl} on higher genus graphs. We focus initially on the computation of the probability of obtaining any single sequence of MBQC measurement outcomes, starting from a surface-code state. Our results are that: (1) In the worst case this can be done with a cost that scales polynomially in the size of the resource, but exponentially in the genus. (2) For any genus $g$ the code space has a basis such that for each basis state the computation is efficient, and (3) There exists an effective state $|\Phi\rangle$ constructed out of the code, the encoded state and the post-measurement unentangled state such that the cost of classically simulating MBQC is exponential in the entanglement of $|\Phi\rangle$. By specializing to a specific family of higher genus graphs and ordering of measurements, we are able to extend these efficiencies to the harder task of sampling from the probability distribution over MBQC outcomes.

The remainder of this paper is organized as follows. In Section II, we define the surface code on tessellations of surfaces of genus $g$. In Section III, we introduce the notions of classical simulation to be used in this paper. In Section IV, we present a method for pointwise evaluating the output distribution of MBQC. In Section V we discuss the efficiency of evaluating partial measurement probabilities, in order to efficiently sample from the output distribution. We conclude in Section VI.

\section{The Surface Code} \label{surfcodesec}
\subsection{Definition}

To define the surface code, we first introduce the notion of a graph embedded on a surface. See \cite{graphsonsurfaces} for a detailed introduction.  In this paper, we consider closed, orientable surfaces $S$ of genus $g$. Given a graph $G=(V,E)$, we say that $G$ is \textit{embedded on} $S$ when $G$ is drawn on $S$ with no edge crossings.  The surface $S$ (minus the image of the embedding) is partitioned by the graph into disjoint regions called \textit{faces}, which are separated from one another by the curves representing edges of $G$. The set of faces is denoted as $F$, and for any $f\in F$, $\partial f$ denotes the \textit{boundary} of $f$, which is the set of edges which separate $f$ from other faces. For any vertex $v\in V$, we let $\delta v$ denote the set of edges that are incident upon $v$ in $G$.  We consider here so-called \textit{cellular embeddings}, which have the property that each face is homeomorphic to an open disk.  For a graph $G$ cellularly embedded on a closed, orientable surface of genus $g$, Euler's formula holds: $|E|-|V|-|F|=2-2g$. When using the term \textit{graph}, we allow for self-loops and redundant edges (what some authors call a \textit{multigraph}), unless explicitly stated otherwise.

Consider a graph $G$ cellularly embedded on an orientable surface $S$: $G=(V,E,F)$, where $G$ is connected. We associate a qubit with each edge $e\in E$.  The surface code is a stabilizer code with stabilizer generators\footnote{Here we use X Pauli operators for the faces and Z for the vertices (as in \cite{rausstoric}), rather than Z operators for the faces and X for the vertices as in most treatments of the surface code. This choice simplifies our discussion. The two code spaces are equivalent up to a global Hadamard transformation.}:
\begin{eqnarray*}
A_v := \prod_{e\in \delta v}Z_e \hspace{.2 in} \forall v \in V,\\
B_f := \prod_{e\in \partial f}X_e \hspace{.2 in} \forall f \in F.
\end{eqnarray*}
\begin{figure}
\includegraphics[width=1.5in]{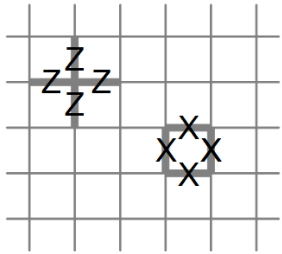}
\caption{\label{fig:stabilizers}The stabilizer group of the surface code is generated by an X-type operator around the boundary of every face, and a Z-type operator on the set of edges incident on every vertex.}
\end{figure}
The code space $\mathcal{CS}$ is defined as the joint +1 eigenspace of all of the stabilizer generators
$$ \mathcal{CS} := \{|\psi\rangle: A_v|\psi\rangle = B_f|\psi\rangle = |\psi\rangle \hspace{.2 in} \forall v \in V, f \in F\}.$$

The stabilizers all commute, because for any $v\in V$ and $f\in F$, $\delta v$ and $\partial f$ always have an even number of edges in common. If $G$ contains any self-loops, then the corresponding edge qubit will be disentangled from the rest for any state $|\psi\rangle \in \mathcal{CS}$, and in the $+1$ X eigenstate. We neglect any such qubit and assume that $G$ contains no self-loops.

For each of the two types of stabilizer generator, any single one can be written as a product of all of the others.  Thus, there are $|V|+|F|-2$ independent, commuting stabilizer generators.  It follows from Euler's formula and the theory of stabilizer codes \cite{nielsenchuang} that the dimensionality of $\mathcal{CS}$ is $2^{2g}$, so the surface code allows for the encoding of $2g$ logical qubits.

\subsection{Encoded Pauli operators}

We now seek $2g$ encoded Pauli X operators $\bar{X}_j$ and encoded Pauli Z operators $\bar{Z}_j$ for $j=1...2g$.  To do so, we shall introduce a few more notions from topological graph theory.  A \textit{cycle} $C$ is a set of edges such that every vertex has an even number of edges incident upon it from $C$ \footnote{Note that some authors require a cycle to be non-null and connected, or contain a maximum of two edges incident on any vertex.  Our definition of cycle also called a \textit{Eulerian subgraph}}. The symmetric difference of any two cycles $C_1$ and $C_2$ is also a cycle, which we shall refer to as the \textit{sum} of $C_1$ and $C_2$. A cycle is called \textit{trivial} if it can be obtained as the sum of the boundaries of some set of faces.   Two cycles are called \textit{homologous} on $G$ if their sum is a trivial cycle. This equivalence relation divides the set of all cycles on $G$ into \textit{homology classes} of mutually homologous cycles. The set of homology classes forms a group under addition, called the first homology group. Each handle in a surface $S$ contributes two independent generators to the first homology group, which is isomorphic to $\mathbb{Z}^{2g}$.  Intuitively, the two generators can be thought of as the cycles that go around the handle, and the cycles that go through it.

An operator of the form $\bar{X} = \prod_{e\in C} X_e$ for any cycle $C$ will commute with all of the stabilizer generators of the surface code.  If $C$ is a trivial cycle, then $\bar{X}$ is equal to a product of some set of $B_f$ operators, and thus acts trivially on the code space.  With this in mind, we define the encoded X operators as $\bar{X}_j = \prod_{e\in C_j} X_e$, where $\{C_j\}$ is a set of $2g$ nontrivial cycles, which are \textit{homologically independent}.  By homologically independent, we mean that no non-trivial linear combination of the cycles $\{C_j\}$ is homologically trivial. This ensures that the $\bar{X}_j$ all act independently on $\mathcal{CS}$ while commuting with the stabilizer generators.
\begin{figure}
\includegraphics[width=3.5in]{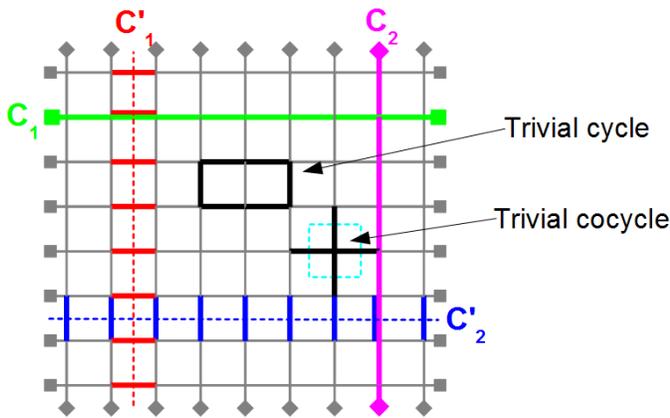}%
\caption{\label{fig:encodedtorus}(Color online) A square toroidal graph ($g=1$), depicted on a plane. The torus is reconstructed by identifying the points marked with diamonds, as well as the points marked with squares. See text for definitions of the cycles and cocycles shown.}
\end{figure}

To define the encoded Pauli Z operators, we use the same construction, but on the dual graph.  For an embedded graph $G=(V,E,F)$, its \textit{dual graph} $\widetilde{G}$ swaps the roles of vertices and faces.  That is, for each face in the original graph we associate a vertex of the dual graph.  Two vertices in $\widetilde{G}$ are then connected by an edge iff the associated faces of $G$ share an edge.  If an edge $e\in E$ is contained entirely within a single face of $G$, rather than separating two distinct faces, then we draw a self-loop in $\widetilde{G}$ for $e$. A cycle $C'$ on $\widetilde{G}$ is called a \textit{cocycle} on $G$, and has the property that $C' \subseteq E: |C' \cap \partial f|=0 \textrm{ (mod 2) } \forall f\in F$. The dual of an embedded graph has a natural embedding on the same surface $S$ as the original graph, where we place each vertex of $\widetilde{G}$ in the center of the associated face of $G$ \cite{graphsonsurfaces}.  Thus, there are also $2^{2g}$ distinct homology classes of cocycles on $G$, where homology is defined with respect to the dual graph embedding. A cocycle on $G$ is trivial if it can be written as $\bigoplus_{v\in \bar{V}}\delta_v$ for some set of vertices $\bar{V} \subseteq V$. We define the encoded Pauli Z operators as $\bar{Z}_j = \prod_{e\in C'_j} Z_e$, where $\{C'_j\}$ is a set of $2g$ homologically independent nontrivial cocycles. To ensure that each encoded X operator anticommutes with the encoded Z operator for the same logical qubit, but commutes with the Z operator for other logical qubits, we must choose the $C_j$ and $C'_j$ such that $|C_j\cap C'_k|=\delta_{jk}$ (mod 2). Figure \ref{fig:encodedtorus} depicts such a choice of ``encoding cycles'' and cocycles for a square toroidal graph.

An algorithm to find a suitable set of cycles $C_k$ and cocycles $C'_k$ satisfying the above criteria - as well as a guarantee of their existence - is provided by the notion of a \textit{tree-cotree decomposition} for an embedded graph, introduced by Eppstein \cite{eppstein}.  For any connected graph $G$, there exists at least one \textit{spanning tree} of $G$, which is defined as a subset of $E$ that forms a tree (is connected and contains no non-null cycles) and visits every vertex in $V$. A spanning tree $T\subseteq E$ contains $|V|-1$ edges.  For any spanning tree $T$, there exists at least one set of edges $C$ within the complement $E\backslash T$ of $T$ in $E$ such that $C$ is a spanning \textit{cotree} of $G$: that is, a spanning tree of the dual graph $\widetilde{G}$. A spanning cotree contains $|F|-1$ edges. For a cellularly embedded graph $G$, Euler's formula implies that the set of leftover edges $X=E\backslash(T\cup C)$ has a cardinality of $2g$. For each edge $e\in X$, the subgraph with edges $T\cup e$ contains exactly one cycle, which we will denote as $T(e)$. Similarly, $C \cup e$ contains exactly one cocycle $C(e)$. If we label the edges in $X$ arbitrarily as $X=\{e_1..e_{2g}\}$ and define $C_j:=T(e_j)$ and $C'_j:=C(e_j)$, then $C_j\cap C'_j = \{e_j\}$ and $C_j \cap C'_k=\emptyset$ for $k\ne j$. The cycles $T(X)$ are also homologically independent (and by corollary likewise for the cocycles $C(X)$) \cite{mohar}.  Thus, a tree-cotree decomposition always provides a suitable definition for the encoded operators of the surface code.

\subsection{The surface-code space}
Now that we have defined encoded qubit operators, we can explicitly construct their eigenstates from the eigenstates of the physical Pauli Z operators.  Let $|x\rangle = |x_1...x_{|E|}\rangle$ for any $|E|$ component bitstring $x$ be an eigenstate of the physical Z operators, with eigenvalue $(-1)^{x_e}$ for the operator $Z_e$. The unique mutual +1 eigenstate of the $2g$ encoded Pauli X operators is
\begin{equation} \label{plusstate}
 |\bar{+}\rangle = |K(G)\rangle := {1\over\sqrt{|E_0(G)|}}\sum_{x\in E_0(G)} |x\rangle,
\end{equation}

\noindent where $E_0(G)$ is the set of bitstrings corresponding to all cycles on G.  We associate bitstrings over $E$ and subsets of $E$ in the natural way: $x_e=1$ iff $e$ is in the subset.  That the state $|\bar{+}\rangle$ is stabilized by all of the $A_v$ operators follows from the fact that since $x$ is a cycle, $A_v|x\rangle = (-1)^{|x\cap \delta v|}|x\rangle = |x\rangle$.  $|\bar{+}\rangle$ is stabilized by all of the $B_f$ operators, because $B_f|x\rangle = |x\oplus \partial f\rangle$, where $\oplus$ denotes mod 2 addition of bitstrings (or equivalently, the symmetric difference of the associated sets).  Since $x\oplus \partial f$ is also a cycle and bitwise addition is invertible, operating on $|\bar{+}\rangle$ by $B_f$ merely permutes the order of the symmetric summation over $E_0(G)$ in Equation \ref{plusstate}. For this same reason, $\bar{X}_j|\bar{+}\rangle=|\bar{+}\rangle$ for all $j=1...2g$.

From the state $|\bar{+}\rangle$, we can construct the rest of the encoded X eigenbasis for $\mathcal{CS}$ by selective application of encoded Z operators. Letting $\alpha$ be any $2g$ component bit string $\alpha_1...\alpha_{2g}$, the state
\begin{equation} \label{xstate}
  |\bar{X}_{\alpha}\rangle := \left(\prod_{j=1}^{2g} \bar{Z}_j^{\alpha_j}\right)|\bar{+}\rangle
\end{equation}

\noindent is the encoded Pauli X eigenstate with eigenvalue $(-1)^{\alpha_j}$ for $\bar{X}_j$. The states $|\bar{X}_{\alpha}\rangle$ provide an orthonormal basis for $\mathcal{CS}$, because
\begin{eqnarray*}
\langle \bar{X}_{\gamma}|\bar{X}_{\alpha}\rangle &=& \langle \bar{+}|\left(\prod_{j=1}^{2g}(\bar{Z}_{j})^{\alpha_j\oplus\gamma_j}\right)|\bar{+}\rangle.
\end{eqnarray*}

If $\gamma_j \ne \alpha_j$ for any $j$, then one can prove that $\langle\bar{X}_{\gamma}|\bar{X}_{\alpha}\rangle=-\langle\bar{X}_{\gamma}|\bar{X}_{\alpha}\rangle=0$ by inserting an $\bar{X}_j$ operator into the above expression and anticommuting it past $\bar{Z}_j$.  If on the other hand $\gamma_j=\alpha_j$ for all $j$, then $\langle\bar{X}_{\gamma}|\bar{X}_{\alpha}\rangle=\langle\bar{+}|\bar{+}\rangle=1$ as expected.

The set $E_0(G)$ appearing in Equation \ref{plusstate} is the so-called \textit{cycle space} of G. From the definition of a cycle and Euler's formula, one can determine the size of the cycle space to be $|E_0(G)|=2^{|E|-|V|+1}=2^{2g+|F|-1}$, assuming that $G$ is connected (see Section \ref{genconsid} for proof). The cycle space of $G$ is a vector space over the binary field $\mathbb{Z}_2$ with a basis composed of all of the face boundaries except one, as well as any set of $2g$ homologically independent nontrivial cycles (such as the $C_j$). 

\section{Classical Simulation of MBQC on Surface-Code States} \label{classicalsim}

In this section, we define our notions of classical simulation of MBQC. A run of MBQC begins with putting in place a resource state $|{\cal{R}}\rangle$, which in the context of the present paper is a state in the code space of a surface code. Subsequently, all qubits in the support of $|{\cal{R}}\rangle$ are measured, where measurement bases are possibly adapted depending on the outcomes of earlier measurements. Finally, the classical output bits, collectively denoted by the vector $\textbf{o}$, are computed as certain parities among measurement outcomes. The probability distribution  for the various values of the output vector $\textbf{o}$ is denoted as $P$, with $P(\textbf{o})$ the probability for the computational outcome $\textbf{o}$.

In this paper, we consider two notions of classically simulating MBQCs, namely
\begin{enumerate}
  \item{Computing the elements $P(\textbf{o})$ of $P$ exactly, for arbitrary output values $\textbf{o}$.}
  \item{Sampling from the probability distribution $P$.}
\end{enumerate}
Consider the scenario where either a measurement-based quantum computer or a classical device simulating it is hidden behind a wall, and one is supposed to figure out the identity of the device merely by looking at its output. 

It is possible to distinguish the real quantum computer from  a classical device efficiently simulating MBQC according to the first notion, e.g. by setting up a problem where $P(\textbf{o})=\delta(\textbf{o},\textbf{m})$, for some $\textbf{m}$; i.e., a needle in a haystack. If the classical device could only compute $P(\textbf{o})$ efficiently for each $\textbf{o}$, it would still generally require exponential time in the length of $\textbf{o}$ to find the correct output $\textbf{m}$. 

However, it is not possible to distinguish a quantum computer from a device efficiently simulating MBQC according to the second criterion, since the probability distribution $P$ fully characterizes the output of the computation. Indeed, the quantum computer itself samples from $P$ \footnote{In  \cite{vdnsim} a distinction is made between `strong' simulations in which certain quantities are computed exactly, and `weak' simulations in which approximations to those quantities are obtained through sampling. In this terminology, the first of the above simulations is a special case of a `strong' simulation and the second simulation is `weak', which may seem counterintuitive after the above. While the second notion of simulation is a weaker in terms of accuracy, it can at least sufficiently closely approximate a wider variety of quantities of interest.}.\medskip

The probability of obtaining a particular sequence of measurement outcomes on all of the $|E|$ qubits is $|\langle \mathcal{R}|\phi\rangle|^2$, where $|\phi\rangle$ is a tensor product of single qubit outcome states. In general, the ability to compute $\langle \mathcal{R}|\phi\rangle$ is sufficient for classical simulation of the first type, since the $P(\textbf{o})$ are all expressible in the form $|\langle \mathcal{R}|\phi\rangle|^2$. Yet, the ability to compute a single such inner product efficiently is not sufficient for efficient classical simulation via sampling from $P$, as the above example illustrates. It is possible however to efficiently sample from $P$ if partial measurement probabilities
$$p\left(|\phi_{\widetilde{E}}\rangle\right)=\textrm{tr}_{\hat{E}}\left(\langle \phi_{\widetilde{E}}|\mathcal{R}\rangle\langle\mathcal{R}|\phi_{\widetilde{E}}\rangle\right)$$ can be computed efficiently. Therein,  $\widetilde{E},\hat {E}$ is a bipartition of the qubits $E$ into a set of measured qubits $\widetilde{E}$ and unmeasured qubits $\hat{E}$, and $|\phi_{\widetilde{E}}\rangle := \otimes_{e\in \widetilde{E}}|\phi_e\rangle$ is a post-measurement state on the measured qubits, representing the outcomes obtained. Consider a step of MBQC where the next qubit to be measured is some $e\in \hat{E}$.  If one now computes $p\left(|\phi_{\widetilde{E}}\rangle\otimes|\phi_e\rangle\right)$, then Bayes' formula yields the probability of obtaining $|\phi_e\rangle$ for qubit $e$, conditioned on the past measurement results:
$$p\left(|\phi_e\rangle\hspace{.05cm}\Bigm \vert \hspace{.05cm}|\phi_{\widetilde{E}}\rangle\right)={{p\left(|\phi_{\widetilde{E}}\rangle\otimes|\phi_e\rangle\right)}\over{p\left(|\phi_{\widetilde{E}}\rangle\right)}}.$$
In this way, one can simulate MBQC by only sampling from distributions over two outcomes, one time for each qubit $e\in E$.  If $p\left(|\phi_{\widetilde{E}}\rangle\right)$ can be computed in a number of steps that scales polynomially in $|E|$, at each step $\widetilde{E}$ of the computation, then the whole simulation can be performed in $poly(\widetilde{E})$ time.

In our context, we will focus initially on the computation of complete local state inner products $\langle \bar{\psi}|\phi\rangle$, where $|\bar{\psi}\rangle$ is a surface-code state, and $|\phi\rangle$ is a product state. We will then find in Section \ref{pmp} that for a certain family of arbitrary genus graphs and a natural ordering of measurements, the task of computing partial measurement probabilities $p\left(|\phi_{\widetilde{E}}\rangle\right)$ reduces to a special case of evaluating $\langle \bar{\psi}|\phi\rangle$.

\section{Product state overlaps and entanglement} \label{overlapentsect}

\subsection{Product state overlaps and the Ising model} \label{overlapsect}

We will begin by showing that the inner product between any surface-code state and an arbitrary product state can be written as a sum of classical Ising model partition functions.  Consider any product state in the physical Hilbert space of the $|E|$ qubits:
$$|\phi\rangle = \bigotimes_{e\in E} \left(a_e|0\rangle_e+b_e|1\rangle_e\right).$$

The inner product between $|\phi\rangle$ and the encoded X eigenstate $|\bar{+}\rangle$ of Equation \ref{plusstate} can be written as a summation over the set $E_0(G)$:
\begin{eqnarray}
\langle \bar{+}|\phi\rangle &=& {1\over\sqrt{|E_0(G)|}}\sum_{x\in E_0(G)} \langle x| \left(\bigotimes_{e\in E} a_e|0\rangle_e+b_e|1\rangle_e\right)\nonumber\\
&=& {1\over\sqrt{|E_0(G)|}}\left(\prod_{e\in E}a_e\right)\sum_{x\in E_0(G)} \prod_{e\in E}\left({b_e\over a_e}\right)^{x_e} \label{overlape0},
\end{eqnarray}
where if $a_e=0$ for any edge $e$ we take a limit as $a_e\rightarrow 0$ and use the continuity of $\langle \bar{+}|\phi\rangle$ as function of the $a_e$ and $b_e$. 

The state overlap in Equation \ref{overlape0} is proportional to the partition function of a classical Ising model with classical spins $\sigma_v \in \{0,1\}$ on the vertices of $G$, and possibly complex couplings $J_e = \tanh^{-1}({b_e\over a_e})$ associated with each edge. It is well known (see \cite{kastgraphtheory} and \cite{fisherdimer}) that the partition function of an Ising model defined on a graph $G$ with couplings $J_e$ can be written as a generating function of cycles on $G$:
\begin{eqnarray}
Z(G,J)&=&2^{|V|}\left(\prod_{e\in E}\cosh(J_e)\right)\mathrm{Cy}(G,\tanh(J))\label{isinge0},
\end{eqnarray}
where 
$$ \mathrm{Cy}(G,w) := \sum_{x\in E_0(G)} \prod_{e\in E} \left(w_e\right)^{x_e}$$
is the generating function of cycles on $G$, where a weight $w_e$ is associated with each edge $e$. Comparing Equations \ref{overlape0} and \ref{isinge0}, we see that if we define the Ising couplings as $J_e := \tanh^{-1}({b_e\over a_e})$, then
\begin{equation} \label{overlapising}
\langle \bar{+}|\phi\rangle = {{\prod_{e\in E}\sqrt{a_e^2-b_e^2}}\over{2^{|V|}\sqrt{|E_0(G)|}}}Z\left(G,J\right).
\end{equation}

Now consider any state $|\bar{\psi}\rangle$ in the surface-code space, with expansion coefficients $c_{\gamma}$ in the encoded X eigenbasis: $|\bar{\psi}\rangle = \sum_{\gamma\in\{0,1\}^{\otimes 2g}} c_{\gamma} |\bar{X}_{\gamma}\rangle$. Expanding the inner product in this basis
\begin{eqnarray} \label{sumoverisings1}
\langle \bar{\psi}|\phi\rangle &=&\sum_{\gamma\in\{0,1\}^{\otimes 2g}} c^*_{\gamma} \langle \bar{+}|\left(\prod_{j=1}^{2g} \bar{Z}_{j}^{\gamma_j}\right)|\phi\rangle.
\end{eqnarray}

Recall that the encoded Pauli Z operators are tensor products of Pauli Z operators acting on the physical qubits. If we take them as operating to the right rather than the left in Equation \ref{sumoverisings1}, then we see that each term is proportional to an inner product between $|\bar{+}\rangle$ and a modified product state $|\phi^\gamma\rangle := \left(\prod_{j=1}^{2g} \bar{Z}_{j}^{\gamma_j}\right)|\phi\rangle$. So we could write Equation \ref{sumoverisings1} as a summation over $2^{2g}$ Ising partition functions, each with different Ising couplings defined from the coefficients of $|\phi^\gamma\rangle$. However, we will find it useful to keep each term in the form of Equation \ref{overlape0}:
\begin{eqnarray} \label{overlapisinggen}
\langle \bar{\psi}|\phi\rangle &=& \mathcal{N}\sum_{\gamma\in\{0,1\}^{\otimes 2g}} c^*_{\gamma}\sum_{x\in E_0(G)} \prod_{e\in E}\left({b^\gamma_e\over a_e}\right)^{x_e}\nonumber \\
&=& \mathcal{N}\sum_{\gamma\in\{0,1\}^{\otimes 2g}} c^*_{\gamma}\mathrm{Cy}(G,w^{\gamma}),
\end{eqnarray}
where $\mathcal{N}:={{\prod_{e\in E}a_e}\over\sqrt{|E_0(G)|}}$ and $b^\gamma_e$ is obtained from $b_e$ by letting $b_e\rightarrow -b_e$ each time the edge $e$ belongs to a cocycle $C'_j$ such that $\gamma_j=1$. The weights $w^\gamma$ are defined as $w^\gamma_e:=b^\gamma_e/ a_e$.

\subsection{Evaluation of product state overlaps}

From Equation \ref{overlapisinggen}, we see that in order to compute an inner product of the form $\langle \bar{\psi}|\phi\rangle$, it is sufficient to be able to evaluate a generating function of cycles on $G$.  Note that the generating function of cycles of a graph $G$ depends only on its vertex and edge sets $V$ and $E$, and makes no reference to an embedding of $G$ on any surface. However, it turns out that embedding $G$ on an orientable surface of genus $g$ allows one to compute $\mathrm{Cy}(G,w)$ in a number of steps that grows exponentially in $g$, while only polynomially in the size of the graph. 

In Appendix \ref{cycles}, we show that for a graph $G$ embedded on an orientable surface of genus $g$, the generating function of cycles on $G$ can be written as
\begin{equation} \label{IsingPFgenform}
\mathrm{Cy}(G,w) = {1\over 2^g}\sum_{\alpha,\beta \in \{0,1\}^{\otimes g}} (-1)^{\alpha\cdot\beta}\textrm{Pf}\left(\mathcal{A}'(w^{\alpha,\beta})\right),
\end{equation}
where $\alpha\cdot \beta$ is the bitwise inner product of the $g$ component bitstrings $\alpha$ and $\beta$, and $\textrm{Pf}\left(\mathcal{A}'(w)\right)$ is the Pfaffian of the weighted adjacency matrix of a modified graph $G'$, which is obtained from the graph $G$ with edge weights $w$. Here, $w^{\alpha,\beta}$ indicates the set of edge weights of $G$ adjusted in a certain way that depends on the bitstrings $\alpha$ and $\beta$. The Pfaffian of a matrix is related to the determinant and is computable in a number of steps that grows polynomially in the size of the matrix. The number of edges of $G'$ is a polynomial in the number of edges of $G$, so $\textrm{Pf}\left(\mathcal{A}'(w^{\alpha,\beta})\right)$ can be computed efficiently in both the number of edges and the genus $g$. Equation \ref{IsingPFgenform} allows for an evaluation of $\mathrm{Cy}(G,w)$ in	$poly(|E|,g)2^{2g}$ steps.

The construction of the adjusted edge weights $w^{\alpha,\beta}$ will be crucial in the following considerations.  In Appendix \ref{cycles}, we define a \textit{canonical encoding scheme}, which is a possible choice of encoding cocycles $C'_k$ defined by cutting and then unfolding the surface $S$ into a topological disk.  In a canonical encoding scheme, the numbering of cocycles $C'_1...C'_{2g}$ is important; in particular, each odd numbered cocycle $C'_{2j-1}$ is paired with an even numbered cocycle $C'_{2j}$.  Given a canonical encoding scheme $C'_1...C'_{2g}$, $w^{\alpha,\beta}_e$ is defined from $w_e$ by multiplying $w_e$ by $-1$ each time $e$ belongs to an odd numbered cocycle $C'_{2j-1}$ such that $\alpha_j=1$, and each time $e$ belongs to an even numbered cocycle $C'_{2j}$ such that $\beta_j=1$.

Consider now the coefficients $c_{\gamma,\rho}$ of an encoded state with respect to a canonical encoding scheme $C'_k$, where $\gamma, \rho \in \{0,1\}^{\otimes g}$, $\gamma_j$ corresponds to the odd numbered cocycle $C'_{2j-1}$, and $\rho_j$ to the even numbered cocycle $C'_{2j}$. Then we may re-write Equation \ref{overlapisinggen} as
\begin{eqnarray}
\langle \bar{\psi}|\phi\rangle &=& \mathcal{N}\sum_{\gamma, \rho\in\{0,1\}^{\otimes g}} c^*_{\gamma, \rho}\mathrm{Cy}(G',w^{\gamma, \rho}) \nonumber.
\end{eqnarray}

The bitstrings $\gamma, \rho$ modify the edge weights $w_e$ here in exactly the same way as the bitstrings $\alpha,\beta$ do in Equation \ref{IsingPFgenform}. So substituting in Equation \ref{IsingPFgenform}:

\begin{eqnarray}
\langle \bar{\psi}|\phi\rangle &=&{\mathcal{N}\over2^g}\sum_{\substack{\alpha,\beta,\gamma,\rho\\ \in\{0,1\}^{\otimes g}}}  c^*_{\gamma,\rho}(-1)^{\alpha\cdot\beta}\textrm{Pf}\left(\mathcal{A}'(w^{\alpha\oplus\gamma,\beta\oplus\rho})\right)\nonumber,
\end{eqnarray}
where $\oplus$ indicates here the binary sum of two bitstrings.  By re-labelling the summation over the dummy indices $\alpha,\beta$, we can rewrite
\begin{eqnarray}
\langle \bar{\psi}|\phi\rangle &=&{\mathcal{N}\over{2^g}}\sum_{\substack{\alpha,\beta,\gamma,\rho\\ \in\{0,1\}^{\otimes g}}} c^*_{\gamma,\rho}(-1)^{(\alpha\oplus\gamma)\cdot(\beta\oplus\rho)}\textrm{Pf}\left(\mathcal{A}'(w^{\alpha,\beta})\right),\nonumber\\
\label{summationover4g}
\end{eqnarray}
where $\mathcal{N}$ is as defined in Section \ref{classicalsim}. Equation \ref{summationover4g} provides a means of computing $\langle \bar{\psi}|\phi\rangle$ in a number of steps that scales as $poly(|E|,g)2^{4g}$.

There exists a family of states in the code space of a surface-code for which the two summations in Equation \ref{summationover4g} cancel each other out, and the complexity of evaluating product state overlaps loses its exponential dependence on $g$.  Consider a state $|\bar{C}^{\delta,\epsilon}\rangle$ parameterized by two g-component bitstrings $\delta, \epsilon$:
\begin{equation} \label{cstatedef} |\bar{C}^{\delta,\epsilon}\rangle := {1\over 2^g}\sum_{\gamma, \rho\in\{0,1\}^{\otimes g}}(-1)^{\delta\cdot\rho+\epsilon\cdot\gamma+\gamma\cdot\rho}|\bar{X}_{\gamma, \rho}\rangle,\end{equation}
where $|\bar{X}_{\gamma, \rho}\rangle$ is the encoded X basis defined by some fixed canonical encoding scheme. It can be verified directly that 
\begin{eqnarray*}
{1\over 2^g}\sum_{\gamma,\rho\in\{0,1\}^{\otimes g}}  (-1)^{\delta\cdot\rho+\epsilon\cdot\gamma+\gamma\cdot\rho}(-1)^{(\alpha\oplus\gamma)\cdot(\beta\oplus\rho)}\\
={1\over 2^g}(-1)^{\alpha\cdot\beta}\sum_{\gamma,\rho\in\{0,1\}^{\otimes g}}  (-1)^{\gamma\cdot(\epsilon\oplus \beta)+\rho\cdot(\delta\oplus\alpha)}
\end{eqnarray*}
equals zero unless $\alpha=\delta$ and $\beta=\epsilon$ component by component, in which case it equals $(-1)^{\delta\cdot\epsilon}2^g$. So, using Equation \ref{summationover4g}:
\begin{eqnarray} \label{specialstateeq}
\langle \bar{C}^{\delta,\epsilon}|\phi\rangle &=&\mathcal{N} (-1)^{\delta\cdot\epsilon}\textrm{Pf}\left(\mathcal{A}'(w^{\delta,\epsilon})\right),
\end{eqnarray}
which can be computed in $poly(|E|,g)$ time. The states $|\bar{C}^{\delta\epsilon}\rangle$ are the encodings of a state that is locally equivalent to a graph state of tensor product form, with one factor per handle. Each handle of the surface $S$ encodes two qubits, and the corresponding graph state is local equivalent to a Bell state; see Figure \ref{GS}. The state $|\bar{C}^{\delta\epsilon}\rangle$ has stabilizers $(-1)^{\delta_j}\bar{X}_{2j-1}\bar{Z}_{2j}$ and $(-1)^{\epsilon_j}\bar{Z}_{2j-1}\bar{X}_{2j}$, for each $j=1...g$. 

\begin{figure}
  \begin{center}
    \includegraphics[width=8cm]{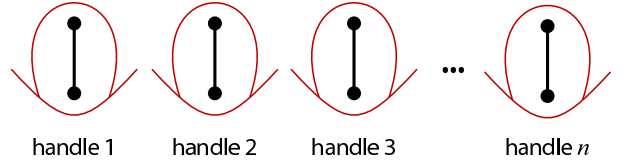}
    \caption{\label{GS} (Color online) The states in the code space for which MBQC remains efficiently simulatable are encodings of graph states. The graph has multiple components,  one per handle. Each handle gives rise to two encoded qubits, and the graph state on each handle is locally equivalent to a Bell state among these two qubits.}
  \end{center}
\end{figure}

The $2^{2g}$ states $|\bar{C}^{\delta,\epsilon}\rangle$ form an orthonormal basis for the code space of the surface code, which can be proven using the orthonormality of the encoded X eigenstates. If the coefficients $\psi_{\delta,\epsilon}$ expanding an arbitrary surface-code state $|\bar{\psi}\rangle$ in the $|\bar{C}^{\delta,\epsilon}\rangle$ basis are known:
$$|\bar{\psi}\rangle=\sum_{\delta,\epsilon\in\{0,1\}^{\otimes g}}\psi_{\delta,\epsilon}|\bar{C}^{\delta,\epsilon}\rangle,$$
then we can improve upon Equation \ref{summationover4g} to compute $\langle \bar{\psi}|\phi\rangle$ in a number of steps that scales as $poly(|E|,g)2^{2g}$:
\begin{eqnarray}
\langle \bar{\psi}|\phi\rangle &=&\mathcal{N}\sum_{\alpha,\beta\in\{0,1\}^{\otimes g}} (-1)^{\alpha\cdot \beta} \psi^{*}_{\alpha,\beta} \textrm{Pf}\left(\mathcal{A}'(w^{\alpha,\beta})\right)\nonumber.\\
\label{summationover2g}
\end{eqnarray}

This observation leads us to the following
\begin{theorem}
\label{Sim1}
Consider an MBQC with generalized flow on a resource surface-code state $|\bar{\psi}\rangle = \sum_{\alpha,\beta \in {\mathbb{Z}_2}^g} \psi_{\alpha,\beta}|\bar{C}^{\alpha,\beta}\rangle$ of $|E|$ qubits, where $g$ is the genus, and the coefficients $\psi_{\alpha,\beta}$ are known. Then, each element $P(\textbf{o})$ of the output probability distribution can be computed exactly in $2^{2g}\mbox{Poly}(|E|,g)$ steps.
\end{theorem}
{\em{Remark:}} A generalized flow consists of a partial ordering among the individual measurement events and a rule for working out which measurement basis depends on which measurement outcome obtained earlier. For a precise definition, see \cite{gflow}. The extra condition of the MBQC possessing a generalized flow does not seem very constraining, since it is the only known condition that guarantees deterministically runnable MBQC.

\begin{proof}
By Theorem 2 of \cite{gflow}, the property of a generalized flow implies strong determinism of the MBQC in question, meaning that each branch of the MBQC is equally likely. We may now split the set $\Omega$ of qubits into two disjoint subsets $O$ and $O^c:=\Omega\backslash O$, where $O^c$ is the set of qubits which condition a correction operation and $O$ the set of qubits which do not. The latter are the output qubits, and can be measured last.

The standard procedure of MBQC with all qubits being measured and the output bits obtained as parities of measurement outcomes is equivalent to the following procedure \cite{onewayqc}: (1) Putting in place the resource state. (2) Performing the local measurements on all qubits $a \in O^c$. (3) Applying Pauli operators on the remaining qubits $b\in O$, conditioned upon the measurement outcomes obtained on the qubits $a \in O^c$. The resulting state of the unmeasured qubits is $|\mbox{out}\rangle_O$. (4) Measuring all qubits $b \in O$. Each measurement outcome yields one bit $o_b$ of output, for all $b \in O$.

By Theorem 2 of \cite{gflow}, the state $|\mbox{out}\rangle_O$, outputted in step 3 of the above procedure, is independent of the measurement outcomes $\textbf{s}|_{O^c}$ of qubits in $O^c$, and all combinations $\textbf{s}|_{O^c}$ of local measurement outcomes  are equally likely. Therefore, it is not necessary to compute each of these probabilities separately. Instead, one may set $\textbf{s}|_{O^c} = \textbf{0}|_{O^c}$. In this case, there are no Pauli corrections on the qubits in $O$. Furthermore,
\begin{equation}
  \label{ovl}
  P(\textbf{o}) = 2^{|O^c|}|\langle \bar{\psi}|\textbf{0}\rangle_{O^c}|\textbf{o}\rangle_O|^2.
\end{equation}  
Therein, $|\textbf{0}\rangle_{O^c}$ is the post-measurement state on the qubits in $O^c$, with every measurement outcome being $s_a=0$ (eigenvalue +1), for all $a \in O^c$. $|\textbf{o}\rangle_O$ is the post-measurement state of the qubits in $O$, with $s_b = o_b$, for all $b \in O$. (In both cases, the basis of the measurement is specified through the algorithm. It is in general not the computational basis.) 

Now, by Eq.~(\ref{summationover2g}), the probability $P(\textbf{o})$ can be computed as a sum over $2^{2g}$ terms. In each term, $\textrm{Pf}\left(\mathcal{A}'(w^{\alpha,\beta})\right)$ can be computed in $\mbox{Poly}(|E|,g)$ steps. $\Box$
\end{proof}

\subsection{Quantum circuit interpretation}
\begin{figure}
  \begin{center}
    \includegraphics[width=8cm]{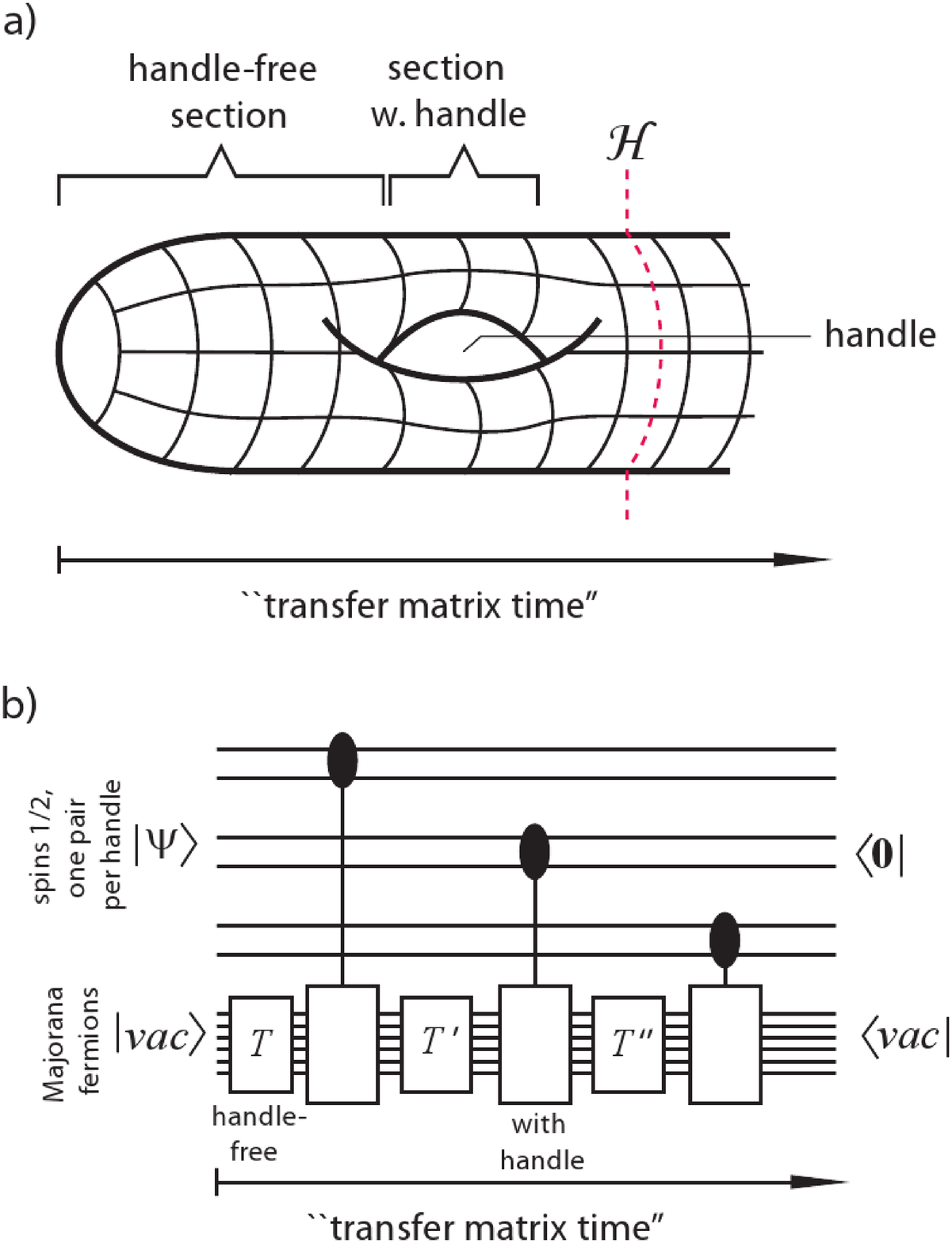}
    \caption{\label{QC} (Color online) Quantum circuit representation of the computation of a surface-code inner-product $\langle \bar{\psi}|\phi\rangle$ for appropriate graphs. In figure b), the encoded state $|\psi\rangle$ is loaded into a register of $2g$ qubits, where it is subsequently entangled with a set of $N$ non-interacting fermions that evolve conditional on the state of the qubits. The interaction is diagonal in the $|C^{\delta,\epsilon}\rangle$ basis of the qubits.}
  \end{center}
\end{figure}

Another perspective on the perhaps surprising efficiency of the states $|\bar{C}^{\delta,\epsilon}\rangle$ comes from thinking of the evaluation of $\langle \bar{\psi}|\phi\rangle$ as equivalent to computing a matrix element of a quantum circuit that entangles a set of $N$ non-interacting fermions to $2g$ qubits (see Figure \ref{QC}). This interpretation is possible in certain situations where the graph $G$ corresponds to a higher genus analog of a rectangular lattice with $N$ rows, such as the punctured cylinder graphs to be introduced in Section \ref{puncturedcylindergraphs}. In this case, the Ising model partition function can be evaluated by a simple generalization of the transfer matrix method.

For an $N\times M$ rectangular lattice, the transfer matrix method
\cite{noneqstatmech} allows the Ising model partition function to be written as the vacuum expectation value of a non-interacting fermion operator on $N$ fermion modes:
$$Z=2^N\langle vac|V^1H^1...V^{M-1}H^{M-1}V^M|vac\rangle$$

where $V^1...V^M$ are known as vertical transfer matrices, expressed as non-interacting fermion operators with parameters that depend on the vertical Ising couplings along a given column of the lattice, while the $H_j$ are non-interacting fermion operators that depend on the horizontal couplings down a given column. This leads to an interpretation of the partition function in terms of a 1D quantum system, where the horizontal dimension acts as time. For certain suitably ``rectangular" non-planar graphs \cite{lgmscthesis}, this formula can be generalized as 

\begin{equation} \label{Zbits}
Z={1\over{2^g}}\sum_{\alpha,\beta\in\{0,1\}^{\otimes g}} (-1)^{\alpha\cdot\beta}\langle vac|\Gamma_{\alpha,\beta}|vac\rangle
\end{equation}

where $\Gamma_{\alpha,\beta}$ is again a product of non-interacting fermion operators, which depend on the bitstrings $\alpha$ and $\beta$ as the virtual time evolution crosses handles in the surface from left to right (see Figure \ref{QC}).

The $2g$ bits $\alpha,\beta$ arise because non-planar vertical boundary conditions (such as those depicted in Figure \ref{QC}) alter the normal mapping from transfer matrices to non-interacting fermion operators via the Jordan-Wigner transformation. To express the product of transfer matrices in terms of non-interacting fermion operators, it is necessary to sum over various parity subspaces of the fermion Fock space, which leads to the summation in Equation \ref{Zbits}. The vertical transfer matrices corresponding to edge qubits directly above a handle take a form $e^{-iJ (-1)^{\hat{n}_k} c_{2k}c_1}=P_k^-e^{iJ c_{2k}c_1}+P_k^+e^{-iJ c_{2k}c_1}$ where $\hat{n}_k$ counts the occupation of the subset of fermion modes $1$-$k$, and $P^{\pm}_{k}$ is the projector into the positive(negative) parity eigenspace of $\hat{n}_k$. The $c_1...c_{2N}$ are Majorana fermion operators and $J$ is a scalar Ising coupling. The parity projectors themselves can each be expanded as $P^{\pm}_{k}={1\over 2}\left(I\pm (-1)^{\hat{n}_k}\right)$, where the action of the operator $(-1)^{\hat{n}_k}$ in the second term turns out to be equivalent to multiplying by $-1$ the horizontal Ising couplings for edges immediately to the left of the handle. Thus in term $\Gamma_{\alpha,\beta}$ of Equation \ref{Zbits}, both $\alpha_j$ and $\beta_j$ are associated with the signs of certain Ising couplings around the $j^{th}$ handle.

When this method for computing the Ising partition function is used for the computation of a surface-code inner-product, we get that
\begin{equation} \label{overlapcircuit} \langle \phi |\bar{\psi}\rangle \propto \langle vac \otimes \mathbf{0}|C\Gamma|vac \otimes \psi\rangle. \end{equation}

where $C\Gamma$ is a ``controlled" fermion operator:
$$C\Gamma:=\sum_{\alpha,\beta\in\{0,1\}^{\otimes g}} |C^{\alpha,\beta}\rangle\langle C^{\alpha,\beta}|\otimes \Gamma_{\alpha,\beta}.$$

Therein, $|\psi\rangle$ is the $2g$-qubit state being encoded into the surface code, and $|\mathbf{0}\rangle$ is the computational basis state on the qubits. Non-interacting fermion operators can be efficiently classically simulated (even when they are non-unitary), so Equation \ref{overlapcircuit} can be evaluated in a number of steps that depends on the number of terms in an expansion of the state $|\psi\rangle$ in the $|C^{\alpha,\beta}\rangle$ basis. In particular, if $|\bar{\psi}\rangle = |\bar{C}^{\delta, \rho}\rangle$ for some $\delta,\rho$, then only one term must be computed and the evaluation of Equation \ref{overlapcircuit} is efficient in all parameters.  For more details on this approach, see \cite{lgmscthesis}.

\subsection{Entanglement in the effective output state}
In the following we will prove tighter bounds on the classical simulation cost on MBQC with surface-code states, in which the exponential factor $2^{2g}$ in Theorem \ref{Sim1} is replaced by smaller exponentials. Specifically, we have
$$
2^{E_{Sch}(|\Phi\rangle)} \leq 2^{n(|\Phi\rangle)} \leq 2^{2g}, 
$$
where $|\Phi\rangle$ is an effective state containing all relevant information about the encoded state $|\psi\rangle$, the encoding and the local bases in which $|\bar{\psi}\rangle$ is measured. Furthermore, $E_{Sch}$ is the Schmidt measure of entanglement and $n$ is the log of the number of terms in a special fixed basis expansion. We have already seen that $n$ can be much smaller than $2g$, namely $n=0$ for the graph states in Fig. \ref{GS}. Our tightest bound involves the Schmidt entanglement measure, and is stated in Theorem \ref{olbtheorem}. A complication arises due to the fact that computing the optimal basis for the Schmidt decomposition in general is a hard problem in itself. In this regard, we show that $E_{Sch}(|\Phi\rangle = n(|\Phi\rangle)$ under mild assumptions; See Theorem \ref{ge2suffthm}.

Recall that the states $|\bar{C}^{\alpha,\beta}\rangle$ can be written in terms of the encoded X-eigenstates of a canonical encoding scheme:
\begin{equation*}
|\bar{C}^{\alpha,\beta}\rangle : = {1\over{2^g}}\sum_{\gamma, \rho \in \{0,1\}^{\otimes g}} (-1)^{\alpha\cdot\rho+\beta\cdot\gamma+\gamma\cdot\rho}|\bar{X}_{\gamma,\rho}\rangle. 
\end{equation*}

It is straightforward to prove that these states are all related to one another by encoded Pauli Z operators for a canonical encoding scheme.  In particular
\begin{eqnarray*}
|\bar{C}^{\alpha,\beta}\rangle&=&(-1)^{\alpha\cdot\beta}\left(\displaystyle{\prod_{j=1}^g}\left(\bar{Z}_{2j-1}\right)^{\alpha_j}\left(\bar{Z}_{2j}\right)^{\beta_j}\right)|\bar{C}^{0,0}\rangle,
\end{eqnarray*}
where $|\bar{C}^{0,0}\rangle$ indicates the state labeled by the g-component zero bitstring for both $\alpha$ and $\beta$, and $\bar{Z}_{k} = \prod_{e\in C'_k}Z_e$.  To simplify notation, define
$$\Psi_{\alpha,\beta}:=(-1)^{\alpha\cdot \beta} \psi^{*}_{\alpha,\beta}.$$

Then we can write any state in the surface-code space as
$$|\bar{\psi}\rangle = \sum_{\alpha,\beta} \Psi_{\alpha,\beta}^* \left(\prod_{k=1}^{g} \bar{Z}_{2k-1}^{\alpha_k}\bar{Z}_{2k}^{\beta_k}\right)|\bar{C}^{00}\rangle.$$
\noindent(Note that we have suppressed the $\in \{0,1\}^{\otimes g}$ under the summation sign to clean up the expression.)

Now consider the quantity $\langle \bar{\psi}|\phi\rangle$.  If we let the Pauli Z operators operate to the right rather than the left we see that
$$\langle \bar{\psi}|\phi\rangle=\langle \bar{C}^{00}|\Phi\rangle,$$
where
\begin{eqnarray}
|\Phi\rangle&:=&\sum_{\alpha,\beta} \Psi_{\alpha,\beta} \left(\prod_{k=1}^{g} \bar{Z}_{2k-1}^{\alpha_k}\bar{Z}_{2k}^{\beta_k}\right)|\phi\rangle  \nonumber \\
&=&\sum_{\alpha,\beta} \Psi_{\alpha,\beta} \left(\prod_{k=1}^{g} \prod_{e\in C'_{2k-1}} Z_e^{\alpha_k} \prod_{e\in C'_{2k}} Z_e^{\beta_k}\right) \bigotimes_{e\in E}|\phi_e\rangle\nonumber.\\  \label{defeff}
\end{eqnarray}

Thus, evaluating the overlap between an arbitrary surface-code state and a product state is equivalent to evaluating the overlap of one of the ``easy'' states $|\bar{C}^{00}\rangle$ with an effective state $|\Phi\rangle$ which is generally \textit{not} a product state of the physical qubits. In a sense, the state $|\Phi\rangle$ reflects an encoding of the $2g$ qubit state $|\psi\rangle$ into the $|E|$ physical qubits of the state $|\phi\rangle$. From Equation \ref{defeff}, it is clear that $|\Phi\rangle$ is a function of: i) the state $|\psi\rangle$ being encoded into the surface code; ii) the chosen encoding scheme $C'_1...C'_{2g}$; and iii) the product state $|\phi\rangle$. In terms of simulating MBQC, the state $|\Phi\rangle$ combines both the specification of the resource state and the particular sequence of measurement outcomes one is computing the probability of (see Section \ref{classicalsim}).

If $|\Phi\rangle$ were to be expanded as a sum over product states, we could evaluate $\langle \bar{C}^{00}|\Phi\rangle$ in a number of steps that grows linearly with the number of terms in the expansion.  The base-2 logarithm of the minimal number of product states that are required to expand a multipartite quantum state is an entanglement monotone known as the \textit{Schmidt measure} \cite{breigelschmidt}. That is, for an N qubit pure state $|\psi\rangle$, the Schmidt measure $E_{Sch}(|\psi\rangle)$ is the minimum number such that
$$ |\psi\rangle = \sum_{j=1}^{2^{E_{Sch}(|\psi\rangle)}} |\chi_1^j\rangle|\chi_2^{j}\rangle...|\chi_N^{j}\rangle$$
for some set of local states $|\chi_k^j\rangle$ for all $j=1...2^{E_{Sch}(|\psi\rangle)}$, $k=1...N$.  We will call the $|\chi_k^j\rangle$ in such an expansion (with $2^{E_{Sch}(|\psi\rangle)}$ terms) an \textit{optimal local basis} for $|\psi\rangle$.  Applying the Schmidt measure to our situation, we immediately have the following result.

\begin{theorem} \label{olbtheorem}
 If an optimal local basis for the effective state $|\Phi\rangle$ is known, then $\langle \bar{\psi}|\phi\rangle$ can be computed in a number of steps that scales as $poly(|E|,g)2^{E_{Sch}(|\Phi\rangle)}$.
\end{theorem}

Computation of $E_{Sch}(|\psi\rangle)$ for a generic multiparty state - no less finding an optimal local basis for it - is generally a very hard problem. Yet an efficient means of computing an optimal local basis is necessary to give Theorem \ref{olbtheorem} much practical significance. In our case, the task of evaluating the Schmidt measure is simplified considerably by the definition of  $|\Phi\rangle$. Since each term Equation \ref{defeff} is a product state, we know that $E_{Sch}(|\Phi\rangle)$ must be less than or equal to $2g$, even though $|\Phi\rangle$ is a state on generally many more than $2g$ qubits.  Furthermore, we can show that under fairly general conditions, Equation in fact \ref{defeff} already provides an optimal local basis for $|\Phi\rangle$.

To state these conditions, we briefly introduce some notation. Let $G-Z$ denote the subgraph of $G$ composed of all edges $e$ such that $|\phi_{e}\rangle$ is not a Pauli Z eigenstate. For any set of edges $\textbf{A}$, let $M_{\textbf{A}}$ be an $|\textbf{A}|\times 2g$ matrix such that $M_{e,k}=1$ if $e\in C'_k$ and $M_{e,k}=0$ if $e\notin C'_k$, for all $e\in\mathbf{A}$. Note that there must exist some edge set $\mathbf{A}=\{e_k\}_{k=1...2g}$ such that $\textrm{rank}\left(M_{\mathbf{A}}\right)=2g$ over the binary field, since the cocycles $C'_k$ are mutually independent as edge sets. For the theorem, we will need to assume a slightly stronger condition:

\begin{theorem} \label{ge2suffthm}
Consider the case where $G-Z$ contains two disjoint sets of $2g$ edges $\mathbf{A}=\{e_k\}_{k=1...2g}$ and $\mathbf{B}=\{e'_k\}_{k=1...2g}$ such that $\textrm{rank}\left(M_{\mathbf{A}}\right)=\textrm{rank}\left(M_{\mathbf{B}}\right)=2g$. Then the expansion in Equation \ref{defeff} yields an optimal local basis for $|\Phi\rangle$ and $E_{Sch}(|\Phi\rangle)=log_2(D)$, where D is the number of nonzero coefficients $\Psi_{\alpha}$.
\end{theorem}

\begin{proof} See Appendix \ref{appproof}.\end{proof}

\begin{figure}
\includegraphics[width=2in]{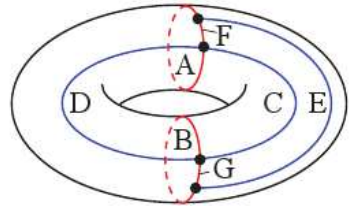}
\caption{\label{fig:torusSM} An embedded graph for which the condition of Theorem \ref{ge2suffthm} does not hold, if one uses cocycles $C'_1=\{C,E\}$, $C'_2=\{A,B\}$, and for at least one edge $e\in \{A,B,C,E\}$ the state $|\phi_e\rangle$ is a Z-eigenstate.}
\end{figure}

The condition assumed for Theorem \ref{ge2suffthm} seems very weak in practice, but in principle may not hold.  Figure \ref{fig:torusSM} shows a simple embedded graph with cocycles $C'_k$ that would violate the condition if any of the cocycle edges were measured in the Z-eigenbasis.

The state $|\Phi\rangle$ and the coefficients $\Psi_{\alpha,\beta}$ can be efficiently computed from the coefficients $\psi_{\alpha,\beta}$ and the definition of the surface-code cocycles $C'_k$.  The number of nonzero $\Psi_{\alpha,\beta}$ is exactly equal to the number of nonzero $\psi_{\alpha,\beta}$. It follows from Theorem \ref{ge2suffthm} then that if the $D$ nonzero coefficients $\psi_{\alpha,\beta}$ are known, and the assumption of the theorem is satisfied, then the quantity $\langle \bar{\psi}|\phi\rangle$ can be evaluated in a number of steps that is polynomial in the size and genus of the embedded graph but increases exponentially with the entanglement in $|\Phi\rangle$, as measured by the Schmidt number. We remark that the assumption of Theorem \ref{ge2suffthm} is satisfied whenever the restriction of each $C'_k$ to $G-Z$ contains two edges that are not shared with any of the other $C'_l$ for $l\ne k$, which would be expected of cocycles on any but the smallest graphs.

\section{Partial Measurement Probabilities} \label{pmp}

In this section we turn to classical simulation in the strong sense of notion 2 in Section \ref{classicalsim}. As a starting point, we recall the result from \cite{rausstoric}, in which it was shown that MBQC on surface-code states can be efficiently simulated when  the underlying graph is planar, and the set of measured qubits $\widetilde{E}$ and its complement $\hat{E}$ are connected at all stages of computation. This result is demonstrated by showing that the probability of obtaining a particular sequence of measurement outcomes on $\widetilde{E}$ is proportional to the inner product between a planar code state on a modified graph $G_{\widetilde{E}} \cup G^*_{\widetilde{E}}$ and a product state. The graph $G_{\widetilde{E}} \cup G^*_{\widetilde{E}}$ is obtained by taking two copies of the subgraph $G_{\widetilde{E}}$ and gluing them together at the boundary of $\widetilde{E}$ and $\hat{E}$.

We will obtain a similar result for general surface-code states, but in the present context the relation is considerably complicated due to the nontrivial topology of $G_{\widetilde{E}} \cup G^*_{\widetilde{E}}$. To handle this new setting, we find it necessary to specialize to cases where underlying graph is what we will call a \textit{punctured cylinder graph} of genus $g$. In doing so, we find that MBQC on a punctured cylinder graph surface code with a natural ordering of single qubit measurements can be simulated (in the strong sense) efficiently in the size of the graph, but inefficiently in $g$. For the states $|\bar{C}^{\alpha,\beta}\rangle$ in this code space, the simulation is completely efficient.

While we expect the result to extend to more general surface-code states and measurement orders, we were unable to prove such as result, and leave it as an open question. Punctured cylinder graphs represent a very simple generalization of the square lattice to higher genus. Furthermore, as a family they contain all higher genus graphs as a graph minor. In principle this makes most of our analysis applicable to arbitrary graphs (see the footnote and discussion of graph minor operations in Section \ref{crossingholessect}), but it is unclear what efficiencies hold in general.

\subsection{General considerations} \label{genconsid}

We consider simulating MBQC by computation of the partial measurement probabilities introduced in Section \ref{classicalsim}:
$$p\left(|\phi_{\widetilde{E}}\rangle\right)=\textrm{tr}_{\hat{E}}\left(\langle \phi_{\widetilde{E}}|\bar{\psi}\rangle\langle\bar{\psi}|\phi_{\widetilde{E}}\rangle\right).$$ 

We will begin by proving some results regarding $p\left(|\phi_{\widetilde{E}}\rangle\right)$ that hold for all connected graphs $G$ with no self loops. In what follows, we will assume as in \cite{rausstoric} that at each stage of the computation, the set of measured edges $\widetilde{E}$ is connected, as is the set of unmeasured edges $\hat{E}$. Our first step will be to construct a Schmidt decomposition for the state $|\bar{+}\rangle$.  This will follow from a few definitions and lemmata.

Let $G(\widetilde{E})$ denote the subgraph of G that contains only the edges $\widetilde{E}$, as well as all vertices $\widetilde{V}$ which have at least one edge incident on them from the set $\widetilde{E}$.  Define $\hat{V}$ and $G(\hat{E})$ similarly, and let $\partial\widetilde{E} \subseteq V := \widetilde{V} \cap \hat{V}$ be the set of vertices containing at least one edge incident upon it from both of the sets $\widetilde{E}$ and $\hat{E}$. We can think of $\partial\widetilde{E}$ as the boundary between the sets $\widetilde{E}$ and $\hat{E}$.

Let $E_0(\widetilde{E})$ denote the set of cycles on the graph $G(\hat{E})$, and define $E_0(\hat{E})$ analogously.  Under the assumption that $\widetilde{E}$ is connected, we have:
\begin{lemma} \label{pml1}
 $|E_0(\widetilde{E})|=2^{|\widetilde{E}|-|\widetilde{V}|+1}$.
\end{lemma}
\begin{proof}
 $E_0(\widetilde{E})$ is a set of $|\widetilde{E}|$ binary variables $\{x_e\}$ satisfying the $|\widetilde{V}|$ binary equations: $\sum_{e\in \delta s} x_e=0$ for all $s\in\widetilde{V}$. If we add together the equation $\sum_{e\in \delta s} x_e=0$ over all $s \in \widetilde{V}$, then each binary variable $x_e$ appears either twice or not at all, and we obtain $0=0$. The equations are otherwise bitwise linearly independent so the total number of independent binary equations is $|\widetilde{V}|-1$. $\Box$
\end{proof}

Now let $E_0(\widetilde{E},\partial\widetilde{E})$ denote the set of binary strings over the edges such that the cycle condition holds everywhere except possibly on the vertices on the boundary $\partial\widetilde{E}$.

\begin{lemma} \label{pml2}
 $|E_0(\widetilde{E},\partial\widetilde{E})|=2^{|\widetilde{E}|-|\widetilde{V}|+|\partial\widetilde{E}|}$.
\end{lemma}
\begin{proof}
 $E_0(\widetilde{E})$ is a set of $|\widetilde{E}|$ binary variables $\{x_e\}$ satisfying the $|\widetilde{V}|-|\partial\widetilde{E}|$ binary equations: $\sum_{e\in \delta s} x_e=0$ for all $s\in\widetilde{V}\backslash \partial\widetilde{E}$. The exclusion of the vertices in $\partial\widetilde{E}$ removes any linearly dependence among these equations. $\Box$
\end{proof}

We now turn to the structure of the set $E_0(\widetilde{E},\partial\widetilde{E})$. For any $\widetilde{x}\in E_0(\widetilde{E},\partial\widetilde{E})$, let $\Delta\widetilde{x}$ be a bitstring encoding the parity of edges from $\widetilde{x}$ incident on the vertices $s\in\partial\widetilde{E}$, i.e. $\Delta\widetilde{x}_s=\sum_{e\in \delta s}\widetilde{x}_e$ for each $s\in\partial\widetilde{E}$. Following Ref \cite{rausstoric}, we call $\Delta\widetilde{x}$ the \textit{syndrome} of $\widetilde{x}$. Then define $E_0(\widetilde{E},u)\subset E_0(\widetilde{E},\partial\widetilde{E})$ to be the set $E_0(\widetilde{E},u):= \{\widetilde{x}\in E_0(\widetilde{E},\partial\widetilde{E}): (\Delta\widetilde{x})_s=u_s \hspace{.1in} \forall s\in \partial\widetilde{E}\}$, where $u = u_1...u_{|\partial\widetilde{E}|}$ is a given syndrome. 

\begin{lemma} \label{pml3}
 $E_0(\widetilde{E},\partial\widetilde{E})=\bigcup_{u\in S} E_0(\widetilde{E},u)$ where $S$ is the set of all bitstrings over the vertices in $\partial\widetilde{E}$ that have an even number of 1's.
\end{lemma}
\begin{proof}
 Since every $\widetilde{x}\in E_0(\widetilde{E},\partial\widetilde{E})$ has some parity $\partial\widetilde{x}$ on the vertices in $\partial\widetilde{E}$, it is immediate that $E_0(\widetilde{E},\partial\widetilde{E})=\bigcup_{u\in \textbf{u}} E_0(\widetilde{E},u)$ for some set $\textbf{u}$ of bitstrings over the vertices in $\partial\widetilde{E}$. We only need to show that $\textbf{u}=S$. Indeed, $E_0(\widetilde{E},u)$ is defined by the $|\widetilde{V}|$ equations: $\sum_{e\in \delta s} \widetilde{x}_e=0$ for all $s\in\widetilde{V}\backslash \partial\widetilde{E}$, and $\sum_{e\in \delta s} \widetilde{x}_e=u_s$ for all $s\in \partial\widetilde{E}$. If we add together these equations for all $s\in \widetilde{V}$, we obtain: $0=\sum_{s\in \partial\widetilde{E}} u_s$. Thus the equations defining $E_0(\widetilde{E},u)$ are inconsistent if $u\notin S$.  On the other hand, there are no further linear dependencies among the equations, so $E_0(\widetilde{E},u)\ne \emptyset$ if $u\in S(\widetilde{E})$. Since $E_0(\widetilde{E},u) \cap E_0(\widetilde{E},u')=\emptyset$ for any $u,u' \in S$ such that $u\ne u'$, it follows that $\textbf{u}$ cannot be a proper subset of $S$. $\Box$
\end{proof}
\begin{corollary} \label{pmc1}
  $|E_0(\widetilde{E},u)|=2^{|\widetilde{E}|-|\widetilde{V}|+1}$ for all $u\in S(\widetilde{E})$ (and similarly for $\hat{E}$).
\end{corollary}
\begin{proof}
The above considerations imply that $E_0(\widetilde{E},u)$ has the same size for each $u\in S(\widetilde{E})$ and so $|E_0(\widetilde{E},\partial\widetilde{E})|=|S(\widetilde{E})||E_0(\widetilde{E},u)|$. From its definition, $|S(\widetilde{E})|=2^{|\partial\widetilde{E}|-|\widetilde{n}|}$, while $|E_0(\widetilde{E},\partial\widetilde{E})|$ is given by Lemma \ref{pml2}. $\Box$
\end{proof}
\begin{corollary} \label{pmc2}
 For any $u\in S(\widetilde{E})$, $E_0(\widetilde{E},u)=\widetilde{z}(u)\oplus E_0(\widetilde{E})$ where $\widetilde{z}(u)$ is any fixed member of the set $E_0(\widetilde{E},u)$.
\end{corollary}
\begin{proof}
 For any $\widetilde{x} \in E_0(\widetilde{E})$ and $\widetilde{z}(u) \in E_0(\widetilde{E},u)$, $\widetilde{x} \oplus \widetilde{z}(u) \in E_0(\widetilde{E})$, since $\Delta(\widetilde{x} \oplus \widetilde{z}(u))=\Delta\widetilde{x}+\Delta\widetilde{z}(u) = u$. Thus, $\widetilde{z}(u)\oplus E_0(\widetilde{E}) \subseteq E_0(\widetilde{E},u)$.  Furthermore, $|\widetilde{z}(u) \oplus E_0(\widetilde{E})|=|E_0(\widetilde{E},u)|$, so $E_0(\widetilde{E},u)=\widetilde{z}(u)\oplus E_0(\widetilde{E})$. $\Box$
\end{proof}

Note that all of the above considerations apply to the edge set $\hat{E}$ as well. We are now in a position to construct a Schmidt decomposition of the state $|\bar{+}\rangle$ with respect to the $(\widetilde{E}, \hat{E})$ bipartition of qubits.  

\begin{theorem}
A Schmidt decomposition of $|\bar{+}\rangle$ is
\begin{equation} \label{schmidt}
|\bar{+}\rangle = {1\over{\sqrt{2^{|\partial\widetilde{E}|-1}}}} \sum_{u\in S} |K_{\widetilde{E}}(u)\rangle\otimes|K_{\hat{E}}(u)\rangle, 
\end{equation}
where
$$|K_{\widetilde{E}}(u)\rangle:={1\over{\sqrt{|E_0(\widetilde{E},u)|}}} \sum_{\widetilde{x}\in E_0(\widetilde{E},u)}|\widetilde{x}\rangle$$
and $|K_{\hat{E}}(u)\rangle$ is defined analogously.
\end{theorem}

\begin{proof}
Note first that
$$E_0(G)=\{(\widetilde{x},\hat{x})\in E_0(\widetilde{E},\partial\widetilde{E})\otimes E_0(\hat{E},\partial\widetilde{E}): \partial\widetilde{x} = \partial\hat{x}\}.$$

Then, by Lemma \ref{pml3}:
\begin{equation} \label{schmidtsecondstep}
|\bar{+}\rangle = {1\over{\sqrt{|E_0(G)|}}} \sum_{\widetilde{u},\hat{u}\in S}\sum_{\widetilde{x}\in E_0(\widetilde{E},\widetilde{u})} \sum_{\hat{x}\in E_0(\hat{E},\hat{u})} \delta_{\widetilde{u},\hat{u}}|\widetilde{x}\rangle|\hat{x}\rangle.
\end{equation}
Equation \ref{schmidt} now follows by the definition of $S$ and working out the normalizations using Corollary \ref{pmc1}. It is easy to see that $\langle K_{\widetilde{E}}(u')|K_{\widetilde{E}}(u)\rangle=\delta_{u,u'}$, and similarly for $\hat{E}$.$\Box$
\end{proof}

The reduced density matrix on the subsystem of qubits corresponding to the edges in $\widetilde{E}$ is then, by Equation \ref{schmidt}:
\begin{equation} \label{reddensitymatrix}
\rho_{\widetilde{E}}=\textrm{tr}_{e\in\hat{E}}\left(|\bar{+}\rangle\langle \bar{+}|\right) = {1\over{2^{|\partial\widetilde{E}|-1}}} \sum_{u\in S} |K_{\widetilde{E}}(u)\rangle \langle K_{\widetilde{E}}(u)|.
\end{equation}

We note that it is evident from the normalization in Equation \ref{reddensitymatrix} that $|\bar{+}\rangle$ obeys the so-called \textit{entanglement area law}: the entropy of entanglement of a block of spins grows linearly with the size of its perimeter.

For an arbitrary surface-code state $|\bar{\psi}\rangle = \sum_{\gamma \in \{0,1\}^{\otimes 2g}} c_{\gamma} |\bar{X}_{\gamma}\rangle$, define $\rho_{\widetilde{E}}(|\bar{\psi}\rangle):=\textrm{tr}_{e\in\hat{E}}\left(|\bar{\psi}\rangle\langle \bar{\psi}|\right)$. Then using Equation \ref{schmidt} we have
\begin{eqnarray}
\rho_{\widetilde{E}}(|\bar{\psi}\rangle)&=&{1\over{2^{|\partial\widetilde{E}|-1}}} \sum_{u,u'\in S} \sum_{\gamma, \delta} c_{\gamma}c^*_{\delta}\bar{Z}^{\gamma}_{\widetilde{E}}|K_{\widetilde{E}}(u)\rangle \langle K_{\widetilde{E}}(u')|\bar{Z}^{\delta}_{\widetilde{E}}\nonumber\\
&&\langle K_{\hat{E}}(u')|\bar{Z}^{\gamma\oplus\delta}_{\hat{E}}|K_{\hat{E}}(u)\rangle\label{rho1},
\end{eqnarray}
where $\bar{Z}^{\gamma}_{\widetilde{E}}:=\prod_{j=1}^{2g}(\prod_{e\in C'_j \cap \widetilde{E}}Z_e)^{\gamma_j}$ and analogously for $\hat{E}$. We can evaluate the matrix product using the definition of $|K_{\widetilde{E}}(u)\rangle$:
\begin{eqnarray}
&\langle K_{\hat{E}}(u')|\bar{Z}^{\gamma\oplus\delta}_{\hat{E}}|&K_{\hat{E}}(u)\rangle={1\over{|E_0(\hat{E},u)|}} \sum_{\hat{x},\hat{y}\in E_0(\hat{E},u)}\langle \hat{y}|\bar{Z}^{\gamma\oplus\delta}_{\hat{E}}|\hat{x}\rangle \nonumber\\
&=&{\delta_{u,u'}\over{|E_0(\hat{E},u)|}}\prod_{j=1}^{2g}(-1)^{(\gamma\oplus\delta)_j|\hat{z}(u)\cap C'_j|}\nonumber\\
&&\sum_{\hat{x}\in E_0(\hat{E})}\prod_{j=1}^{2g}(-1)^{(\gamma\oplus\delta)_j|\hat{x}\cap C'_j|} \label{matrixprod},
\end{eqnarray}
where $\hat{z}(u)$ is any fixed member of the set $E_0(\hat{E},u)$ and we have used Corollary \ref{pmc2} in the last step. Consider any value of $j$ such that $(\gamma\oplus\delta)_j=1$.  If there exists any $\hat{y}\in E_0(\hat{E})$ such that $|\hat{y}\cap C'_j|=1$ (mod 2), then the above summation over $\hat{x}\in E_0(\hat{E})$ vanishes.  That is because for each $\hat{x}\in E_0(\hat{E})$, the bitstring $\hat{x}\oplus \hat{y}$ term will have the opposite sign as the $\hat{x}$ term and the two will cancel, since $|\hat{x}\oplus\hat{y}\cap C'_j|=1+|\hat{x}\cap C'_j|$. Let $A$ denote the set of $j\in \{1...2g\}$ such that there exists a $\hat{y}\in E_0(\hat{E})$ satisfying $|\hat{y}\cap C'_j|=1$ (mod 2). Let $B$ denote the set of $j$ that are not in $A$, but for which $C'_j \cap \hat{E}\ne \emptyset$. So we can rewrite the RHS of Equation \ref{matrixprod} as
\begin{eqnarray*}
{\delta_{u,u'}}\prod_{j\in A}\delta_{\gamma_j,\delta_j}\prod_{j\in B}(-1)^{(\gamma\oplus\delta)_j|\hat{z}(u)\cap C'_j|}
\end{eqnarray*}
because if $\gamma_j=\delta_j$ for all $j\in B$ then each term in the summation over $\hat{x}\in E_0(\hat{E})$ is positive, cancelling the overall factor of ${|E_0(\hat{E},u)|}^{-1}$. Using this and Equation \ref{rho1}, we can now consider a partial measurement probability for the qubits in $\widetilde{E}$:
\begin{eqnarray}
p\left(|\phi_{\widetilde{E}}\rangle\right) &=& \langle\phi_{\widetilde{E}}|\rho_{\widetilde{E}}(|\bar{\psi}\rangle)|\phi_{\widetilde{E}}\rangle \nonumber\\
&=& {1\over{2^{|\partial\widetilde{E}|-1}}}  \sum_{\gamma, \delta} c_{\gamma}c^*_{\delta}\prod_{j\in A}\delta_{\gamma_j,\delta_j}\langle\phi_{\widetilde{E}}\otimes \phi^*_{\widetilde{E}}|\bar{Z}^{\gamma}_{\widetilde{E}_1}\bar{Z}^{\delta}_{\widetilde{E}_2}\nonumber\\
&&\sum_{u\in S}(-1)^{\sum_{j\in B}(\gamma\oplus\delta)_j|\hat{z}(u)\cap C'_j|}|K_{\widetilde{E}}(u)\otimes K_{\widetilde{E}}(u)\rangle\nonumber.\\ \label{prob}
\end{eqnarray}

In this notation, we have replaced the product of two matrix elements $\langle\phi_{\widetilde{E}}|\bar{Z}^{\gamma}_{\widetilde{E}}|K_{\widetilde{E}}(u)\rangle \langle K_{\widetilde{E}}(u')|\bar{Z}^{\delta}_{\widetilde{E}}|\phi_{\widetilde{E}}\rangle$ in the Hilbert space of $|\widetilde{E}|$ qubits with a single matrix element $\langle\phi_{\widetilde{E}}\otimes \phi^*_{\widetilde{E}}|\bar{Z}^{\gamma}_{\widetilde{E}_1}\bar{Z}^{\delta}_{\widetilde{E}_2}|K_{\widetilde{E}}(u)\otimes K_{\widetilde{E}}(u')\rangle$ in the Hilbert space of $2|\widetilde{E}|$ qubits. Here $|\phi^*_{\widetilde{E}}\rangle$ is a product state obtained from $|\phi_{\widetilde{E}}\rangle$ by complex conjugating $a_e$ and $b_e$ for each $e\in \widetilde{E}$.  $\bar{Z}^{\gamma}_{\widetilde{E}_1}$ is the operator $\bar{Z}^{\gamma}_{\widetilde{E}}$ applied to the first copy of $\widetilde{E}$, denoted as $\widetilde{E}_1$ (and likewise for $\widetilde{E}_2$).

Equation \ref{matrixprod} relates $p\left(|\phi_{\widetilde{E}}\rangle\right)$ to a summation over states in the Hilbert space of a surface code on the graph $G(\widetilde{E}_1)\cup G(\widetilde{E}_2)$, defined by taking two copies of $G(\widetilde{E})$ and gluing them together at the vertices in the boundary $\partial \widetilde{E}$ (as in \cite{rausstoric}). When the set $B$ is empty for example, the ket in Equation \ref{matrixprod} becomes $\sum_{u\in S}|K_{\widetilde{E}}(u)\otimes K_{\widetilde{E}}(u)\rangle$, which is the logical +1 X eigenstate $|\bar{+}\rangle$ of the surface code on $G(\widetilde{E}_1)\cup G(\widetilde{E}_2)$.  This is because the set of cycles $E_0(G(\widetilde{E}_1)\cup G(\widetilde{E}_2))$ on this graph has the structure: $E_0(G(\widetilde{E}_1)\cup G(\widetilde{E}_2))=\{(\widetilde{x},\widetilde{y})\in E_0(\widetilde{E},\partial\widetilde{E})\otimes E_0(\widetilde{E},\partial\widetilde{E}): \partial\widetilde{x} = \partial\widetilde{y}\}$. In the notation of Equation \ref{plusstate}:
\begin{eqnarray*}
|K(G(\widetilde{E}_1)\cup G(\widetilde{E}_2))\rangle={1\over\sqrt{2^{|\partial\widetilde{E}|-1}}}\sum_{u\in S}|K_{\widetilde{E}}(u)\otimes K_{\widetilde{E}}(u)\rangle.
\end{eqnarray*}

\subsection{MBQC on punctured cylinder graphs} \label{puncturedcylindergraphs}

We now define the family of punctured cylinder graphs and apply the above analysis to them. To construct a cellularly embedded punctured cylinder graph, consider an $N\times M$ square lattice with periodic boundary conditions in the vertical direction, embedded on the surface of a solid disk. Then, imagine drilling $g$ thin holes (or ``slots'') through the disk, each one in between two rows of vertices on the graph. Finally, vertical edges are extended through each slot, as in Figure \ref{fig:puncturedcylinder} below. 
\begin{figure}[!h]
\centering
\includegraphics[width=3in]{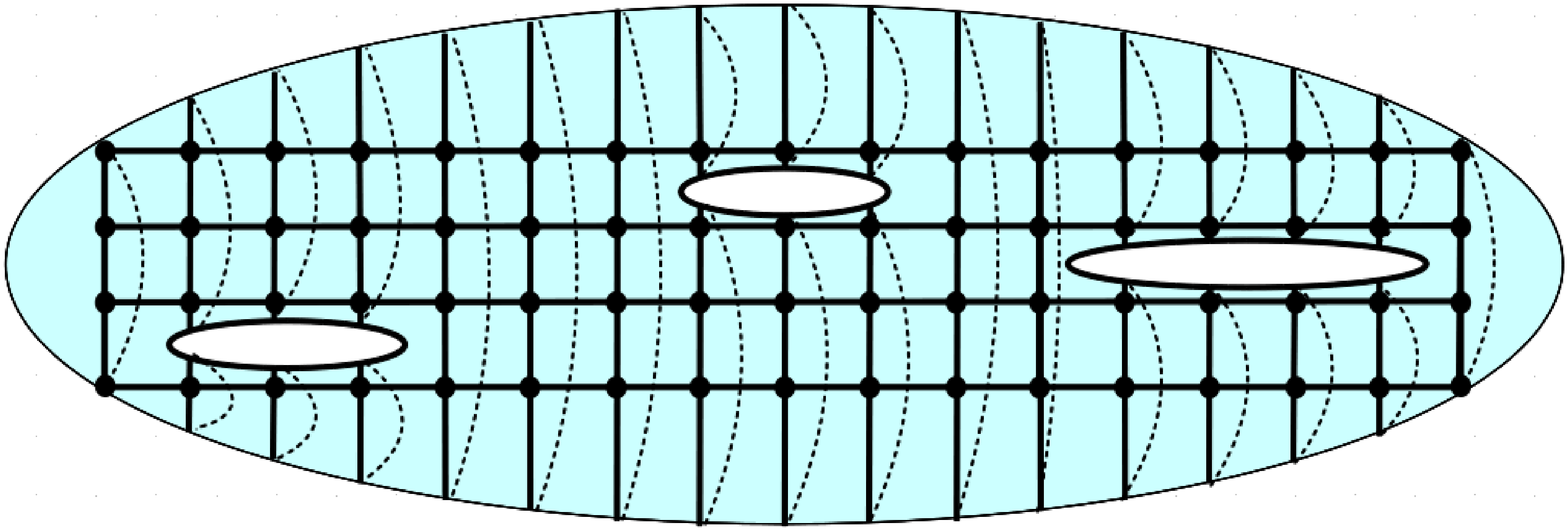}
\caption{(Color online) A three-slot punctured cylinder graph cellularly embedded on a surface of genus three.}
\label{fig:puncturedcylinder}
\end{figure}

A family of such punctured cylinder graphs is parameterized by the dimensions of the lattice along with the position and width of each slot: $\{N,M,\{x_1,y_1,K_1\}, ... \{x_g,y_g,K_g\}\}$.  We will take the slots to be ordered from left to right ($x_{j+1}>x_{j}$), and assume that no two slots are above one another ($x_{j+1} \ge x_j+K_j$). A flattened representation of a punctured cylinder graph is shown in Figure \ref{circuitgraphfig}.
\begin{figure}[!h]
\centering
\includegraphics[width = 3.5in]{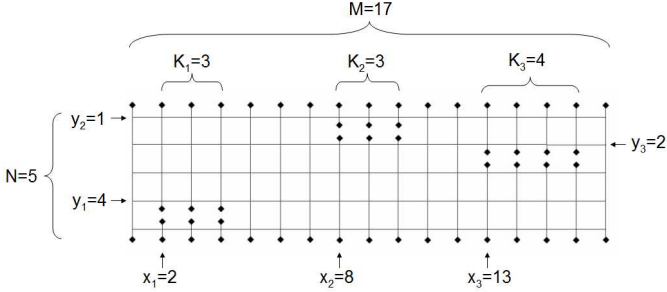}
\caption{A three-slot punctured torus graph. The $n=3$ handles have positions $(x_j,y_j)$ and widths $K_j$ for $j=1...n$. Pairs of points marked by diamonds are identified within each column, according to Figure \ref{fig:puncturedcylinder}.}
\label{circuitgraphfig}
\end{figure} 

It will be necessary to give a concrete set of encoding cocycles $C'_k$ for the punctured cylinder surface code. A suitable choice is shown in Figure \ref{homcircuitgraphfig}. These cocycles in fact constitute a canonical encoding scheme (as defined in Appendix \ref{cycles}). This is because one can continuously deform the loops drawn in Figure \ref{homcircuitgraphfig} for the cocycles $C'_k$ such that they form a canonical polygonal schema (this does not change the edge sets $E_k$ defined in Appendix \ref{cycles}). This deformation is shown in Figure \ref{fig:schema} for the simple case of a double torus. 

\begin{figure}[!h]
\centering
\includegraphics[width = 3in]{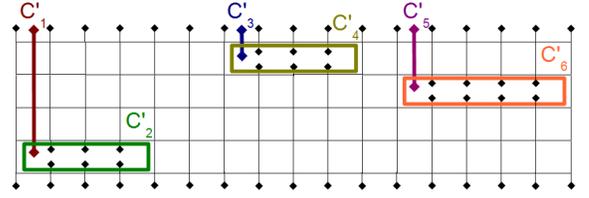}
\caption{(Color online) A set of 6 non-trivial cocycles $C'_k$ on a three-slot punctured cylinder graph. The edges included in a cocycle are those that are crossed by the line depicted.}
\label{homcircuitgraphfig}
\end{figure}

We will also assume a particular order in which to make the single qubit measurements of MBQC on the punctured cylinder lattice, in order to simplify the analysis.  Since the punctured cylinder graph has a left and right boundary, we may unambiguously start at the leftmost column, and measure the qubits column by column proceeding to the right.  That is: first we measure all qubits on the vertical edges in column 1, then all of the qubits on horizontal edges between columns 1 and 2, then the vertical edge qubits in column 2, and so on.  We further take the measurements to occur row by row as one moves down a column of horizontal or vertical edges. For brevity, we will call this ordering of measurements \textbf{LtoR}. \textbf{LtoR} seems to be a natural choice because it mimics the simple temporal order in a quantum circuit, and it satisfies the assumption of the previous section that both $\widetilde{E}$ and $\hat{E}$ are connected at all stages.

Our main result of this section is the following theorem:
\begin{theorem} \label{mbqctheorem}
Consider a state $|\bar{\psi}\rangle$ in the surface-code space of a punctured cylinder graph $G$ of genus $g$. For MBQC on $|\bar{\psi}\rangle$ with the measurement ordering \textbf{LtoR}, at any step $\widetilde{E}$ of computation and for product state of outcomes $|\phi_{\widetilde{E}}\rangle$:
$$p\left(|\phi_{\widetilde{E}}\rangle\right) = \alpha\langle \phi(G'(\widetilde{E}))|\bar{\psi}(G'(\widetilde{E}))\rangle,$$
where $G'(\widetilde{E})$ is an embedded punctured cylinder graph of genus less than or equal to $2g$, $|\bar{\psi}(G'(\widetilde{E}))\rangle$ is a state in the codespace of 
$G'(\widetilde{E})$, $|\phi(G'(\widetilde{E}))\rangle$ is a product state, and $\alpha$ is a known proportionality.
\end{theorem}
\begin{proof} See Appendix \ref{appholes}.
\end{proof}
Together with Theorem \ref{olbtheorem}, we then have the following Corollary:
\begin{corollary} \label{mbqcsimcorr}
For MBQC with the measurement scheme \textbf{LtoR} on a punctured cylinder code state $|\bar{\psi}\rangle$, if an optimal local basis for the effective state $|\Phi\rangle$ corresponding to the inner product in Theorem \ref{mbqctheorem} is known at each step $\widetilde{E}$ of computation, then the probability distribution $P$ over the outcomes of the next measurement can be classically sampled from in $poly(|E|,g)2^{E_{Sch}(|\Phi\rangle)}$ steps.
\end{corollary}

As a special case of Theorem \ref{mbqcsimcorr}, MBQC on the states $|\bar{C}^{\alpha,\beta}\rangle$ in the codespace of the punctured cylinder code can be simulated completely efficiently in the strong sense of sampling:

\begin{theorem} \label{mbqcefftheorem}
The probability distribution $P$ of computational output values of MBQC on one of the states $|\bar{C}^{\alpha,\beta}\rangle$ in the code space of the punctured cylinder code (with the measurement scheme \textbf{LtoR}) can be sampled from efficiently in both $|E|$ and $g$.
\end{theorem}
\begin{proof} See Appendix \ref{specialstatesapp}.
\end{proof}

\section{Conclusion}

We have considered the classical simulation of MBQC with surface-code states as resource states.  We first showed that for surface-code states the probability of obtaining any single MBQC outcome can be computed in a number of steps that scales polynomially in the size of the surface-code embedded graph, and at worst exponentially in its genus.  We found a family of states in the code space of any surface code for which this probability can be computed efficiently in both the size and the genus of the graph. For intermediate cases, we found a connection between the complexity of computing such probabilities and entanglement. In particular, the cost scales exponentially in the Schmidt measure of a state which combines the specification of MBQC outcomes and the quantum state being encoded into the surface code.  We also considered the task of sampling from the probability distribution over MBQC outcomes, and saw that for MBQC on a certain family of embedded graphs with a simple ordering of measurements, this task is equivalent to computing a single MBQC outcome probability for a modified graph. From this we were able to define a class of higher genus surface-code states for which MBQC can be efficiently classically simulated.

\subsection*{Acknowledgments}
We are thankful to Pradeep Sarvepalli and Shuhang Yang for useful discussions. This work is supported by NSERC, CIFAR, Mprime and IARPA.

\appendix

\section{Evaluating the Generating Function of Cycles} \label{cycles}

In this Appendix we show that the generating function of cycles on an embedded graph $G$ can be written in the form of Equation \ref{IsingPFgenform}:
\begin{equation*}
\mathrm{Cy}(G,w) = {1\over 2^g}\sum_{\alpha,\beta \in \{0,1\}^{\otimes g}} (-1)^{\alpha\cdot\beta}\textrm{Pf}\left(\mathcal{A}'(w^{\alpha,\beta})\right).
\end{equation*}

To arrive at Equation \ref{IsingPFgenform}, we map the problem of evaluating the generating function of cycles on $G$ to the problem of evaluating the generating function of perfect matchings on a modified graph $G'$. Then we apply a result from \cite{loeblpfaff} to evaluate this generating function.

A \textit{perfect matching} $M$ of a graph $G$ is a subset of the edges of $G$ such that every vertex contains exactly one edge incident upon it in $M$. Let $\mathcal{PM}(G)$ denote the set of all perfect matchings on $G$.  If to each edge $e$ we associate a weight $w_e$, then the generating function of perfect matchings on $G$ is defined as
$$ P(G,w) := \sum_{M\in \mathcal{PM(G)}} \prod_{e\in M} w_e.$$ We now define a modified embedded graph $G'$ by the following four rules, adapted from \cite{fisherdimer} and \cite{barahona}:
\begin{itemize}
\item If any vertex $v$ has exactly one edge incident upon it, remove it and the incident edge from $G$
\item For any vertex $v$ with exactly two edges $a$ and $b$ incident upon it, split $v$ into two vertices connected by a new edge with weight 1, as shown in Figure \ref{fig:rules}a.
\item For any vertex $v$ with exactly three edges incident upon it, replace $v$ with six vertices and nine edges as shown in Figure \ref{fig:rules}b.
\item For any vertex with $n>3$ edges incident upon it, first replace $v$ with $n-1$ vertices of degree three as shown in Figure \ref{fig:rules}c, and then follow the rule for a degree three vertex for each of the resulting vertices.
\end{itemize}

\begin{figure}
\includegraphics[width=3.5in]{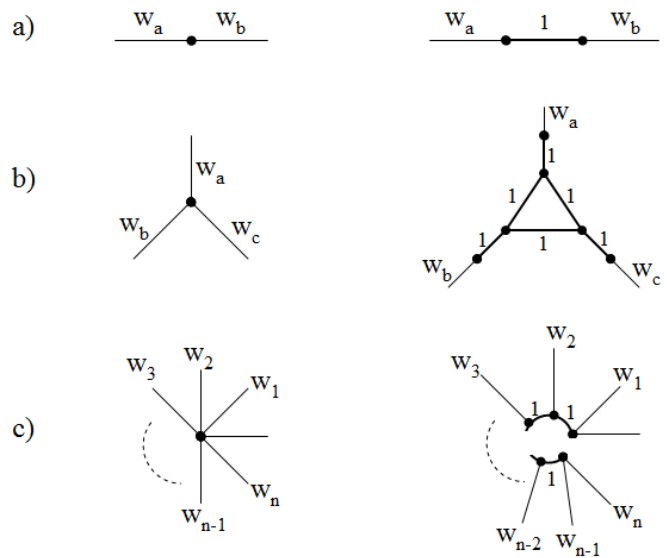}%
\caption{Transformations at each vertex from $G$ to a modified graph $G'$.}\label{fig:rules}
\end{figure}

The above rules define a graph $G'$ which differs from $G$ only locally around each vertex (and deletion of vertices of degree one).  Thus, it also has a natural embedding on $S$ where the modification around each vertex can be made arbitrary small. Furthermore, it can be verified that there exists a one-to-one mapping between cycles $x\in E_0(G)$ on $G$ and perfect matchings $M\in \mathcal{PM}(G')$ on $G'$, and that the product of edge weights for a given cycle on $G$ is equal to that of the associated perfect matching on $G'$. So: $$\mathrm{Cy}(G,w)=P(G',w'),$$ where $w'$ denotes the edge weights $w_e$ for all $e\in E$ along with $w_e=1$ for all of the new edges introduced in the transformation $G\rightarrow G'$.

In \cite{loeblpfaff}, Galluccio and Loebl study the problem of evaluating the generating function of perfect matchings on a graph $G'$ that is embedded on an orientable surface of genus $g$.  Their main result (Theorem 3.9 of \cite{loeblpfaff}) is a formula for $P(G',w)$ that can be written in the form of Equation \ref{IsingPFgenform}.  Therein the function $\textrm{Pf}\left(\mathcal{A}'(w^{\alpha,\beta})\right)$ is the Pfaffian of a $|V'|\times|V'|$ weighted adjacency matrix $\mathcal{A}'(w^{\alpha,\beta})$, where $|V'|$ is the number of vertices in the graph $G'$ (a few more edges may need to be added to $G'$, as we shall see at the end of this section). The Pfaffian $\textrm{Pf}(M)$ of a $2N \times 2N$ matrix $M$ is a polynomial in the matrix entires that is related to the determinant, and can be computed in $poly(N)$ time. In their work, Galluccio and Loebl take the embedded graph as being specified by a so-called \textit{canonical polygonal schema}. A \textit{curve} in $S$ is a continuous map $h: [0,1] \rightarrow S$, and a \textit{loop} is a curve with $h(1)=h(0)$. A canonical polygonal schema of a graph $G$ is obtained from its embedding on $S$ by cutting $S$ along $2g$ loops $\mathcal{C}_1...\mathcal{C}_{2g}$, chosen such that after the cutting $S$ can be unfolded into a convex polygon $B_0$ with $4g$ sides. Each cut $\mathcal{C}_k$ produces two paired sides of $B_0$, which we denote as $\mathcal{C}^1_k$ and $\mathcal{C}^2_k$, and the sides of $B_0$ are arranged in clockwise order as $\mathcal{C}^1_1,\mathcal{C}^1_2,\mathcal{C}^2_1,\mathcal{C}^2_2,\mathcal{C}^1_3...\mathcal{C}^2_{2g}$. The closed surface $S$ can be reconstructed by glueing $\mathcal{C}^1_k$ and $\mathcal{C}^2_k$ back together with the proper orientation. 

\begin{figure}[t]
\includegraphics[width=2.5in]{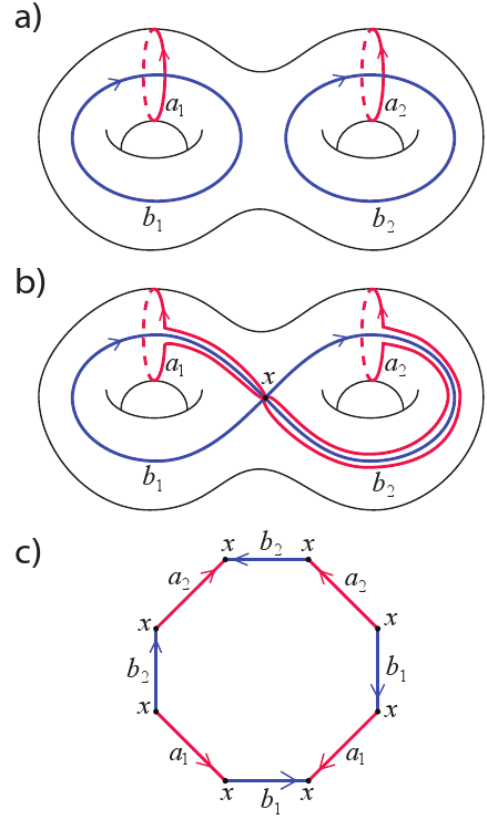}
\caption{\label{fig:schema}(Color online)  a) A set of loops on a double torus defining a set of punctured cylinder graph cocycles $C'_k$ as shown in Figure \ref{homcircuitgraphfig}, b) a deformation of these loops to define a canonical polygonal schema, and c) the octagon obtained after cutting along the loops shown in b) and unfolding the surface. To visualize the move between Figures b) and c), one may separate the two tori and imagine unfolding each individually into a single torus with a boundary, as shown in Figure 84 of \cite{springer} (p.71).}
\end{figure}

To use the results of reference \cite{loeblpfaff} then, we require a suitable set of loops $\mathcal{C}_1...\mathcal{C}_{2g}$ on $S$. These can be chosen as follows: draw $2g$ non-self-intersecting curves on $S$ that all begin and end at a common \textit{base point} $x$, but are otherwise non-overlapping and non-crossing, and such that for each $j$: $\mathcal{C}_{2j-1}$ goes around the $j^{th}$ handle, and $\mathcal{C}_{2j}$ goes though the $j^{th}$ handle. See Figure \ref{fig:schema}b for an example. Consider now the original graph $G$ embedded on $S$. Choose the basepoint to be at the center of some face $f$ of $G$.  Without loss of generality, we may choose the $\mathcal{C}_k$ to avoid the vertices of $G$ and cross the embeddings of the edges of $e$ only at isolated points. After cutting $S$ along these loops, we are left with a plane graph plus some cut edges. We define $G_0$ as the plane graph on $B_0$ consisting of all of the vertices of $G$ and all of the edges that do not cross any of the cuts $\mathcal{C}_k$. Let $E_k$ denote the set of edges of $G$ that cross the loop $\mathcal{C}_k$ an odd number of times.  We now prove a few properties of the sets $E_k$.

\begin{lemma} \label{lemmacocycle}
For each $k\in \{1...2g\}$, the edge set $E_k$ is a cocycle of G.
\end{lemma}
\begin{proof}
For each $f\in F$, the cycle $\partial f$ defines a loop or set of disjoint loops $\mathcal{C}_f$ on $S$. Since $\mathcal{C}_f$ forms the boundary of a region of $S$, the loops $\mathcal{C}_f$ and $\mathcal{C}_k$ cross an even number of times (this follows from the Jordan Curve Theorem).  So, there cannot be an odd number of edges $e\in \partial f$ that cross $\mathcal{C}_k$ an odd number of times. Thus $|\partial f\cap E_k|$ is even for every face $f\in F$. $\Box$
\end{proof}

\begin{lemma} \label{lemmahomol}
The cocycles $E_k$ are homologically independent on $\widetilde{G}$.
\end{lemma}
\begin{proof}
If this were not true, then for some collection $Y\subseteq \{1...2g\}$ of the $E_k$ and some set $\widetilde{V}\subset V$ of vertices, we would have:
$$\bigoplus_{k\in Y} E_k = \bigoplus_{v\in\widetilde{V}} \delta v.$$ The edge set $\bigoplus_{k\in Y} E_k$ is precisely the set of edges that are crossed an odd number of times by the loop $\mathcal{C}_Y$, which we define as the concatenation of the loops $\mathcal{C}_k$ for all $k\in Y$, in some arbitrary order. By continuously deforming $\mathcal{C}_Y$ around the vertices $v\in\widetilde{V}$, one obtains a modified loop $\hat{\mathcal{C}}_Y$ that crosses all edges $e\in E$ either an even number of times or not at all. Removal of the loop $\mathcal{C}_Y$ from $S$ does not separate the surface, because cutting along \textit{all} of the $\mathcal{C}_k$ results in a single polygon $B_0$, which is still connected. Since $\hat{\mathcal{C}}_Y$ is related to $\mathcal{C}_Y$ by a continuous deformation, its removal does not separate $S$ either. But now we can prove a contradiction, because a non-surface-separating loop must intersect at least one edge of $G$ an odd number of times. 

To demonstrate this, we use a result from \cite{moharcabello} (cf. Lemma 3). First, we define an embedded graph $\hat{G}_Y$ which combines the original graph $G$, and the loop $\hat{\mathcal{C}}_Y$ as follows. Add a vertex to $G$ at each point where $\hat{\mathcal{C}}_Y$ crosses an edge of $G$, and a vertex at the base point $x$ of the canonical polygonal schema. For each section of $\hat{\mathcal{C}}_Y$ between two intersection points with $G$, add an edge that traces the section.  Finally, add edges that trace $\hat{\mathcal{C}}_Y$ between the basepoint $x$ and the points where $\hat{\mathcal{C}}_Y$ first crosses an edge from $x$. The new edges that trace out the loop $\hat{\mathcal{C}}_Y$ define a cycle of $\hat{G}_Y$, which we denote as $\hat{c}_Y$. For any edge $e$ of $G$ that was split into several edges $e_1...e_k$ by the transformation $G\rightarrow \hat{G}_Y$, let $\hat{e}$ denote the set $\{e_1..e_k\}$. The number of times that $\hat{\mathcal{C}}_Y$ crosses the edge $e$ of $G$ is then $|\hat{e}|-1$. Let $\{Q_{1}...Q_{2g}\}$ denote any set of $2g$ homologically independent cycles on $G$, and for each $Q_j$ let $\hat{Q}_j$ denote the corresponding cycle on $\hat{G}_Y$ (simply let $\{e\}\rightarrow \hat{e}$ for any edge $e$ that is crossed by $\hat{\mathcal{C}}_Y$). By Lemma 3 of \cite{moharcabello}, there exists some $j$ such that $\hat{Q}_j$ is crossed by $\hat{c}_Y$ an odd number of times, iff $\hat{c}_Y$ is a homologically non-trivial cycle on $\hat{G}_Y$. The cycle $\hat{c}_Y$ must be homologically non-trivial on $\hat{G}_Y$, because if it were not then it would form the boundary of a set of faces of $\hat{G}_Y$, and cutting along $\hat{c}_Y$ (or equivalently $\hat{\mathcal{C}}_Y$) would separate the surface $S$ (a similar argument shows that the $\hat{Q}_j$ are homologically independent on $\hat{G}_Y$, which is necessary for our use of the result in \cite{moharcabello}). So $\hat{c}_Y$ crosses $\hat{Q}_j$ an odd number of times, for some $j$. But, if there were no edge $e$ of $G$ that was crossed an odd number of times by $\hat{\mathcal{C}}_Y$, then $\hat{Q}_j$ and $\hat{c}_Y$ could only cross an even number of times (or zero). So there does exists such an edge $e$. $\Box$
\end{proof}

\begin{theorem} \label{encthm}
The cocycles $E_k$ constitute a possible choice of encoding cocycles $C'_k$ for the surface code on $G$.
\end{theorem}
\begin{proof}
By Lemma \ref{lemmahomol}, the cocycles $E_k$ are homologically independent on $\widetilde{G}$. All that's left is to show that with encoded Z cocycles defined as $C'_k := E_k$, there exists at least one set of encoding cycles for the X operators $C_k$ on $G$ such that $|C_j\cap C'_k| = \delta_{jk}$ (mod 2). As discussed in Section \ref{surfcodesec}, a tree-cotree decomposition of $G$ guarantees the existence of homologically independent cycles $T(e_1)...T(e_{2g})$ and homologically independent cocycles $C(e_k)$...$C(e_{2g})$ on $G$ such that $|T(e_j)\cap C(e_k)| = \delta_{jk}$. The cocycles $C(e_k)$ along with the edge sets $\delta_v$ for all $v\in V$ form a basis for all cocycles on $G$ with respect to the symmetric difference of sets. So, $C'_k = \bigoplus_{m\in Y_k} C(e_m) \oplus \bigoplus_{v\in V_k}\delta_v$ for some $Y_k \subseteq \{1...2g\}$ and $V_k \subseteq V$. Since the $C'_k$ are homologically independent, the $2g\times 2g$ matrix $A$ defined by $A_{mk} \in \{0,1\}: A_{mk}=1 \textrm{ iff } m\in Y_k$ is invertible over the binary field $\mathbb{Z}_2$. Let $A^{-1}$ denote its $\mathbb{Z}_2$ inverse and define the set $Y^{-1}_j$ as the set of all $l$ for which $A^{-1}_{jl}=1$. Then define a set of encoding cycles as $C_j := \bigoplus_{l\in Y^{-1}_j} T(e_l)$. Using the definition of a cycle and $|T(e_l)\cap C(e_m)| = \delta_{lm}$ $$|C_j\cap C'_k| = \bigoplus_{l\in Y^{-1}_j} \bigoplus_{m\in Y_k} \delta_{lm} = \bigoplus_{m=1}^{2g} A^{-1}_{jm}A_{mk} = \delta_{jk},$$ where in this expression $\bigoplus$ denotes mod 2 addition of numbers. Finally, the cycles $C_j$ so defined are homologically independent on $G$ because the matrix $A^{-1}$ is invertible over $\mathbb{Z}_2$. $\Box$
\end{proof}

\begin{definition}
Given a canonical polygonal schema $\{\mathcal{C}_k\}$, a \textit{canonical encoding scheme} is the choice of encoding cocycles $C'_k:=E_k$. This is a valid one by Theorem \ref{encthm}.\newline
\end{definition} 

So far, we've defined a canonical polygonal schema $\{\mathcal{C}_k\}$ for $S$, and the associated canonical encoding scheme $\{E_k\}$ for the surface code of $G$. We now apply these concepts to the modified graph $G'$. Since all of the vertices of $G$ belong to the interior of $B_0$, we can perform the graph modification $G\rightarrow G'$ in an arbitrarily small neighborhood of each vertex \textit{after} unfolding the embedded graph $G$. We take the $\mathcal{C}_k$ to be chosen such that they avoid crossing any edge $e$ that is incident on a vertex of degree one (one may merely drag $\mathcal{C}_k$ across that vertex to avoid $e$). This yields a canonical polygonal schema for $G'$, where the edge set $E_k$ is still the set of edges of $G'$ that cross the cut $C_k$ an odd number of times.

Another modification of the graph $G'$ is necessary for us to use Equation \ref{IsingPFgenform} (see Corollary 3.9 of \cite{loeblpfaff}). Consider any edge $e$ that crosses $n$ possibly non-distinct cuts $\mathcal{C}_{k_1}...\mathcal{C}_{k_n}$, in that order as you follow $e$ in one direction. If $n>1$, then one modifies $G'$ by adding $2n$ vertices and replacing $e$ by a string of edges $e_1...e_{2n+1}$ connected in a chain such that $e_{2j-1}$ crosses one cut $\mathcal{C}_{k_j}$ for each $j=1...n$. Edge $e_1$ is given weight $w_e$ while the rest of the edges receive a weight of $w_{e_j}=1$. Call this transformation \textit{bridge splitting}. Bridge splitting guarantees that no edge of $G'$ crosses more than one cut, or any single cut more than once. Let $E'_k$ denote the set of edges of $G'$ that cross the cut $\mathcal{C}_{k}$. Let $w'$ continue to denote the set of weights of the edges of $G'$. One may verify that the generating function $P(G',w')$ of perfect matchings is unchanged by bridge splitting. After bridge splitting, a few more minor transformations of the graph may be necessary (see \cite{loeblpfaff}), but these do not affect our analysis.

Now we consider the construction of the weighted adjacency matrices $\mathcal{A}'(w^{\alpha,\beta})$ in Equation \ref{IsingPFgenform}. Let $G'_0$ be the subgraph of $G'$ that belongs entirely to $B_0$.  $G'_0$ contains all of the vertices of $G'$, and all of the edges that do not cross any cut. An \textit{orientation} of a graph is an assignment of a direction to each edge.  As a plane graph, it can be shown that $G'_0$ has an orientation $D_0$ of its edges such that the boundary of each face has an odd number of edges oriented clockwise \cite{kaststatistics}.  Such an orientation is called a \textit{basic} orientation, and we fix a particular one $D_0$. For each $k \in \{1...2g\}$, Gallucio and Loebl show that $G'_0 \cup E'_k$ has a natural plane embedding, and a unique orientation $D_k$ of the edges $E_k$ such that $(D_0,D_k)$ is a basic orientation in this plane embedding. For any $\alpha, \beta \in \{0,1\}^{\otimes 2g}$, a so-called \textit{relevant orientation} of $G'$ is defined as follows: start with the orientation $(D_0,D_1,D_2...D_{2g})$, and reverse the orientation of all edges in $E'_{2k-1}$ if $\alpha_k=1$, and reverse the orientation of all edges $E'_{2k}$ if $\beta_k=1$, for each $k=1...g$. For any two vertices $u$,$v$ of $G'$, we define the matrix element $[A'(w')^{\alpha,\beta}]_{u,v}$ to be $0$ if $u$ and $v$ are not connected by an edge, $w'_e$ if $u$ and $v$ are connected by an edge $e$ oriented from $u$ to $v$, and $-w'_e$ if $u$ and $v$ are connected by an edge $e$ oriented from $v$ to $u$, where the edge orientations are defined by the relevant orientation $\alpha,\beta$.\newline
\indent The matrix $A'(w')^{\alpha,\beta}$ depends both on $\alpha$ and $\beta$ and the edge weights $w'_e$. Reversing the orientation of an edge has the same effect as multiplying the corresponding edge weight by $-1$. So, we may write $A'(w')^{\alpha,\beta}=A'(w'^{\alpha,\beta})$ where $A'(w')$ denotes the adjacency matrix $A'(w')^{0,0}$ of $G'$ corresponding to the concatenation of the basic orientations $(D_0,D_1,D_2...D_{2g})$, and $w'^{\alpha,\beta}$ is the set of edge weights $w'$ after we multiply by $-1$ all edge weights along the cocycle $E'_{2k-1}$ if $\alpha_k=1$ and along the cocycle $E'_{2k}$ if $\beta_k=1$. Recall that the edge weights $w'$ of $G'$ are determined by the edge weights $w$ of $G$, so we could denote $A'(w')$ as $\mathcal{A}'(w)$, where the matrix $\mathcal{A}'(\cdot)$ incorporates the effect of the graph modifications $G\rightarrow G'$. We will now show that $\textrm{Pf}\left(A'\left(w'^{\alpha,\beta}\right)\right)=\textrm{Pf}\left(\mathcal{A}'(w^{\alpha,\beta})\right)$, where $w^{\alpha,\beta}$ is the set of edge weights $w$ of $G$ after we multiply by $-1$ the edge weight $w_e$ once for each time it belongs to a cocycle $E_{2k-1}$ for which $\alpha_k=1$, and once for each time it belongs to a cocycle $E_{2k}$ for which $\beta_k=1$.  Each nonzero term of the Pfaffian $\textrm{Pf}(A'(w')^{\alpha,\beta})$ depends on $w'$ only via the product of edge weights $w'^{\alpha,\beta}_e$ for the edges $e$ in a particular perfect matching of $G'$ (see Definition 1.3 in \cite{loeblpfaff}). For any edge $e\in E$ that was replaced by a set of edges $e_1...e_{2n+1}$ during the bridge splitting process, a perfect matching of $G'$ contains either none or all of $\{e_1,e_3...e_{2n+1}\}$. If $e\in E_k$, then there are an odd number of $e_{2j-1}$ that cross the cut $\mathcal{C}_k$. Multiplying the weights of all of these edges by $-1$ yields an overall minus sign for a term containing $\{e_1,e_3...e_{2n+1}\}$, which has the exact same effect as letting $w_e \rightarrow -w_e$ before bridge splitting.  If on the other hand $e$ crosses $\mathcal{C}_k$ but an even number of times, then there are an even number of $e_{2j-1}$ that cross the cut $\mathcal{C}_k$, and there is no effect on $\textrm{Pf}(A'(w')^{\alpha,\beta})$ from multiplying the weights of these edges by $-1$. Finally, with $\textrm{Pf}\left(A'\left(w'^{\alpha,\beta}\right)\right)=\textrm{Pf}\left(\mathcal{A}'(w^{\alpha,\beta})\right)$, Equation \ref{IsingPFgenform} holds up to a possible overall minus sign by Theorem 3.9 of \cite{loeblpfaff}. The possible minus sign depends upon $D_0$ and the structure of the graph $G'$, but not on the edge weights $w^{\alpha,\beta}$. So we may neglect it as it would only add an overall phase to $\langle \bar{\psi}|\phi\rangle$ in Equation \ref{overlapisinggen}.

\section{Proof of Theorem \ref{ge2suffthm}} \label{appproof}

We will show that under the assumptions of the theorem, if $$|\Phi\rangle = \sum_{j=1}^{s} |\chi_1^{j}\rangle|\chi_2^{j}\rangle...|\chi_{|E|}^{j}\rangle$$ for any set of single qubit states $|\chi_k^{j}\rangle$, then $s\ge D$. Our first step will be to be to isolate a single term of Equation \ref{defeff} by taking a partial inner product between $|\Phi\rangle$ and a particular state on the qubits in \textbf{A}.

In the following, the distinction between the even and odd numbered cocycles will not be important, so we simplify notation by writing the coefficients $\Psi_{\alpha,\beta}$ as $\Psi_{\alpha}$ where $\alpha$ is now a $2g$ component bitstring.  Then we can rewrite Equation \ref{defeff} as:
\begin{eqnarray*}
|\Phi\rangle
&=&\sum_{\alpha \in \{0,1\}^{\otimes 2g}} \Psi_{\alpha} \left(\prod_{k=1}^{2g} \prod_{e\in C'_{k}} Z_e^{\alpha_k}\right) \bigotimes_{e\in E}|\phi_e\rangle\nonumber\\
&=&\sum_{\alpha \in \{0,1\}^{\otimes 2g}} \Psi_{\alpha} \bigotimes_{e\in E}\left(Z_e\right)^{{[M\alpha]}_e}|\phi_e\rangle,
\end{eqnarray*}

where $M$ is the $|E|\times 2g$ matrix such that $M_{e,k}=1$ if $e\in C'_k$ and $M_{e,k}=0$ if $e\notin C'_k$, for all $e\in E$.  $[M\alpha]_e:=\sum_{k=1}^{2g} M_{e,k}*\alpha_k$.

Write $|\phi_{e}\rangle = a_e|0\rangle + b_e|1\rangle$ for any edge $e$.  Now we define $|\phi^{0,\perp}_{e}\rangle := b^*_e|0\rangle + a^*_e|1\rangle$, and $|\phi^{1,\perp}_{e}\rangle := b^*_e|0\rangle - a^*_e|1\rangle$.  It is easy to verify that for any edge $e$ and binary variable $\gamma_k \in \{0,1\}$ $$\langle\phi^{\gamma_k,\perp}_{e}|\left(Z_e\right)^{\alpha_k}|\phi_e\rangle = \delta_{\alpha_k,\gamma_k}2a_eb_e.$$

In particular, $|\phi^{1,\perp}_{e}\rangle$ is perpendicular to $|\phi_e\rangle$ for any edge $e$, while $|\phi^{0,\perp}_{e}\rangle$ is perpendicular to $Z_e|\phi_e\rangle$ for any edge $e$. First we write Equation \ref{defeff} in the form
\begin{eqnarray}
|\Phi\rangle
&=&\sum_{\alpha \in \{0,1\}^{\otimes 2g}} \Psi_{\alpha}|\phi_{rest}^{\alpha}\rangle\otimes|\phi_{\mathbf{A}}^{\alpha}\rangle\otimes|\phi_{\mathbf{B}}^{\alpha}\rangle\label{defeff2},
\end{eqnarray}
where $|\phi_{rest}^{\alpha}\rangle$ is a $\alpha$-dependent product state on all of the qubits in the complement of $\mathbf{A}\cup\mathbf{B}$ in $E$, and
$$|\phi_{\mathbf{A}}^{\alpha}\rangle:=\bigotimes_{k=1}^{2g}\left(Z_{e_k}\right)^{{[M\alpha]}_{e_k}}|\phi_{e_k}\rangle=\bigotimes_{k=1}^{2g}\left(Z_{e_k}\right)^{{[M_{\mathbf{A}}\alpha]}_{k}}|\phi_{e_k}\rangle.$$

The states $|\phi^{\gamma_k,\perp}_{e_k}\rangle$ for any $2g$ component bitstring $\gamma$ can now be used to pick out a single term in Equation \ref{defeff2}, because
$$\left(\bigotimes_{k=1}^{2g} \langle \phi^{\gamma_k,\perp}_{e_{k}}|\right)|\phi_{\mathbf{A}}^{\alpha}\rangle=\left(\prod_{k=1}^{2g} 2a_{e_k}b_{e_k}\right)\delta_{\gamma,[M_{\mathbf{A}}\alpha]}$$
and thus
\begin{eqnarray}
\left(\bigotimes_{k=1}^{2g} \langle \phi^{[M_{\mathbf{A}}\gamma]_k,\perp}_{e_{k}}|\right)|\Phi\rangle &=& \Psi_{\gamma}\left(\prod_{k=1}^{2g} 2a_{e_k}b_{e_k}\right)|\phi_{rest}^{\gamma}\rangle\otimes|\phi_{\mathbf{B}}^{\gamma}\rangle \nonumber.\\\label{s1}
\end{eqnarray}

The only value of $\alpha$ for which $[M_{\mathbf{A}}\alpha]=[M_{\mathbf{A}}\gamma]$ is $\alpha=\gamma$, because by assumption the square matrix $M_{\mathbf{A}}$ has full rank and hence is invertible. Since $|\phi_{e_k}\rangle$ is not a Z-eigenstate, $2a_{e_k}b_{e_k}$ is nonzero for each k. We can show that the states $\{|\phi_{rest}^{\gamma}\rangle\otimes|\phi_{\mathbf{B}}^{\gamma}\rangle\}$ for various bitstrings $\gamma$ are a linearly independent family of states.  This follows from the assumption of the second set $\mathbf{B}$ of non-Z eigenstate edges $\{e'_{k}\}$ for which $M_{\mathbf{B}}$ has full rank.  For we can repeat the above trick to show that each $|\phi_{rest}^{\gamma}\rangle\otimes|\phi_{\mathbf{B}}^{\gamma}\rangle$ has a component that is perpendicular to subspace spanned by the rest of the $|\phi_{rest}^{\gamma}\rangle\otimes|\phi_{\mathbf{B}}^{\gamma}\rangle$:
 \begin{displaymath}
   \left(\langle \phi_{rest}^{\gamma}|\otimes\bigotimes_{k=1}^{2g} \langle \phi^{[M_{\mathbf{B}}\gamma]_k,\perp}_{e'_{k}}|\right)|\phi_{rest}^{\alpha}\rangle\otimes|\phi_{\mathbf{B}}^{\alpha}\rangle \left\{
     \begin{array}{lr}
       =0 & : \alpha \ne \gamma\\
       \ne 0 & : \alpha = \gamma
     \end{array}
   \right..
\end{displaymath} 

The RHS is zero if $\alpha \ne \gamma$, but is a nonzero vector if $\alpha=\gamma$. So the state $|\phi_{rest}^{\gamma}\rangle\otimes|\phi_{\mathbf{B}}^{\gamma}\rangle$ has a component that lies along the vector $|\phi_{rest}^{\gamma}\rangle\otimes\left(\bigotimes_{k=1}^{2g} |\phi^{[M_{\mathbf{B}}\gamma]_k,\perp}_{e'_{k}}\rangle \right)$, but all of the other $|\phi_{rest}^{\alpha}\rangle\otimes|\phi_{\mathbf{B}}^{\alpha}\rangle$ are orthogonal to it. Thus $|\phi_{rest}^{\gamma}\rangle\otimes|\phi_{\mathbf{B}}^{\gamma}\rangle$ cannot be written as a linear combination of the others, for each $\gamma$.

Now let $|\Phi\rangle = \sum_{j=1}^{s} |\chi_1^{j}\rangle|\chi_2^{j}\rangle...|\chi_{|E|}^{j}\rangle$ be any other expansion of $|\Phi\rangle$ into some number $s$ of product states. Write it as
$$|\Phi\rangle = \sum_{j=1}^{s} |\chi_{E\backslash \mathbf{A}}^{j}\rangle\otimes|\chi_{\mathbf{A}}^{j}\rangle.$$
Then
\begin{eqnarray*} 
\left(\bigotimes_{k=1}^{2g} \langle \phi^{[M_{\mathbf{A}}\gamma]_k,\perp}_{e_{k}}|\right)|\Phi\rangle = \sum_{j=1}^s\left(\left(\bigotimes_{k=1}^{2g} \langle \phi^{[M_{\mathbf{A}}\gamma]_k,\perp}_{e_{k}}|\right)|\chi_{\mathbf{A}}^{j}\rangle\right)\\
|\chi_{E\backslash \mathbf{A}}^{j}\rangle.
\end{eqnarray*}

Comparing this with Equation \ref{s1}, we see that for each $\gamma$ for which $\Psi_{\gamma}$ is nonzero, $|\phi_{rest}^{\gamma}\rangle\otimes|\phi_{\mathbf{B}}^{\gamma}\rangle$ can be written as a linear combination of the $s$ states $|\chi_{E\backslash \mathbf{A}}^{j}\rangle$. Let $D$ be the number of such nonzero $\Psi_{\gamma}$. Since each $|\phi_{rest}^{\gamma}\rangle\otimes|\phi_{\mathbf{B}}^{\gamma}\rangle$ is linearly independent, there must be enough states $|\chi_{E\backslash \mathbf{A}}^{j}\rangle$ to span a $D$ dimensional space.  So, $s\ge D$. Since this applies to any decomposition of the form $|\Phi\rangle = \sum_{j=1}^{s} |\chi^1_{j}\rangle|\chi^2_{j}\rangle...|\chi^{|E|}_{j}\rangle$, we conclude that $E_{Sch}(|\Phi\rangle)=log_2{D}$. $\Box$

\section{Proof of Theorem \ref{mbqctheorem}} \label{appholes}

Specializing to punctured cylinder codes and the measurement ordering \textbf{LtoR} allows us to greatly simplify Equation \ref{prob}. We consider two separate cases in turn.

\subsection{Measurements between holes} \label{mmtbetweenholes}

We say that MBQC is ``between'' two holes when for some $k$, all of the edges in column $x_k+K_k$ are in the set $\widetilde{E}$, while all edges in column $x_{k+1}$ are still in the set $\hat{E}$. In this subsection we will show that

\begin{lemma} Theorem \ref{mbqctheorem} holds when computation is between holes.\end{lemma}

\begin{proof} With the encoding cocycles $C'_k$ chosen as depicted in Figure \ref{homcircuitgraphfig}, then the set $A$ from Equation \ref{prob} contains all of the values from $2k+1..2g$, and the set $B$ is empty. Furthermore, $C'_k$ lies entirely within the edge set $\widetilde{E}$ for $k\le 2k$. Then Equation \ref{prob} becomes
\begin{eqnarray}
p\left(|\phi_{\widetilde{E}}\rangle\right) &=& {1\over{2^{|\partial\widetilde{E}|-1}}}  \sum_{\gamma, \delta} c_{\gamma}c^*_{\delta}\prod_{j=2k+1}^{2g}\delta_{\gamma_j,\delta_j}\nonumber\\
&&\langle\phi_{\widetilde{E}}\otimes \phi^*_{\widetilde{E}}|\bar{Z}^{\gamma}_{\widetilde{E}_1}\bar{Z}^{\delta}_{\widetilde{E}_2}\sum_{u\in S}|K_{\widetilde{E}}(u)\otimes K_{\widetilde{E}}(u)\rangle\nonumber.\\ \label{probbtwholes}
\end{eqnarray}

In section \ref{genconsid}, we saw that the state $\sum_{u\in S}|K_{\widetilde{E}}(u)\otimes K_{\widetilde{E}}(u)\rangle$ is the logical +1 X eigenstate $|\bar{+}\rangle$ associated with a surface code on the effective graph $G(\widetilde{E}_1)\cup G(\widetilde{E}_2)$. In this setting, graph $G(\widetilde{E}_1)\cup G(\widetilde{E}_2)$ has a natural embedding on a surface of genus $2k$, where the first $k$ holes come from the subgraph $G(\widetilde{E}_1)$ and the second $k$ holes come from the subgraph $G(\widetilde{E}_1)$. The set of $4k$ encoding cocycles for a surface code on $G(\widetilde{E}_1)\cup G(\widetilde{E}_2)$ can be chosen to be $C'_1...C'_{2k}$ on the edges $\widetilde{E}_1$, along with $C'_1...C'_{2k}$ on the edges $\widetilde{E}_2$. Then, the state $$\bar{Z}^{\gamma}_{\widetilde{E}_1}\bar{Z}^{\delta}_{\widetilde{E}_2}|K(G(\widetilde{E}_1)\cup G(\widetilde{E}_2))\rangle$$
is precisely the encoded X eigenstate $|X_{\gamma_1...\gamma_{2k},\delta_1...\delta_{2k}}\rangle$ in the surface-code space for $G(\widetilde{E}_1)\cup G(\widetilde{E}_2)$. If we furthermore define
\begin{eqnarray}
\widetilde{c}_{\gamma_1...\gamma_{2k},\delta_1...\delta_{2k}} := \sum_{\substack{\gamma_{2k+1}...\gamma_{2g}\\ \delta_{2k+1}...\delta_{2g}\\ \in \{0,1\}}} c_{\gamma_1...\gamma_{2g}}c^*_{\delta_1...\delta_{2g}}\prod_{j=2k+1}^{2g}\delta_{\gamma_j,\delta_j} \nonumber, \\ \label{cbardef1}
\end{eqnarray} 
then the probability of a outcome on the edges in $\widetilde{E}$ from the original graph is exactly proportional to an inner product with a state in the code space of the surface code on $G(\widetilde{E}_1)\cup G(\widetilde{E}_2)$:
\begin{eqnarray}
p\left(|\phi_{\widetilde{E}}\rangle\right) &=& {1\over{\sqrt{2^{|\partial\widetilde{E}|-1}}}}\langle\phi_{\widetilde{E}}\otimes \phi^*_{\widetilde{E}}|\nonumber\\
&& \sum_{\substack{\gamma_{1}...\gamma_{2k}\\ \delta_{1}...\delta_{2k}}}\widetilde{c}_{\gamma_1...\gamma_{2k},\delta_1...\delta_{2k}}|X^{G'(\widetilde{E})}_{\gamma_1...\gamma_{2k},\delta_1...\delta_{2k}}\rangle\nonumber.\\ \label{probbtwholesfinal}
\end{eqnarray}

Here $\widetilde{c}$ is an effective tensor of coefficients in the encoded X-basis for a state in the surface-code space of $G'(\widetilde{E}):=G(\widetilde{E}_1)\cup G(\widetilde{E}_2)$. The inner product between this state and the product state $|\phi_{\widetilde{E}}\otimes \phi^*_{\widetilde{E}}\rangle$ yields the partial measurement probability. This confirms Theorem \ref{mbqctheorem} for the cases when computation is between holes. Now, we turn to the other stages of MBQC on a punctured cylinder code state.\end{proof}

\subsection{Measurements crossing holes} \label{crossingholessect}

If the boundary $\partial \widetilde{E}$ contains vertices in a column between $x_k$ and $x_k+K_k$ for any $k$, then some acrobatics are required to keep Equation \ref{prob} expressible in the simple form of Equation \ref{probbtwholes}. This scenario occurs as the computation ``crosses holes'' from left to right on the lattice $G$. Figure \ref{fig:holeclasses} below shows the part of a punctured cylinder graph $G$ around the $k^{th}$ hole. In particular, we will focus on the measurement steps after edge $a$ in Figure \ref{fig:holeclasses} has been measured, but before edge $b$ is measured. Before edge $a$ is measured, or after edge $b$ is measured, the situation is no more complicated than when computation is ``between holes''. In this subsection we will show that nevertheless,

\begin{lemma} Theorem \ref{mbqctheorem} holds when computation is crossing holes.\end{lemma}

\begin{figure}[!h]
\centering
\includegraphics[width=3in]{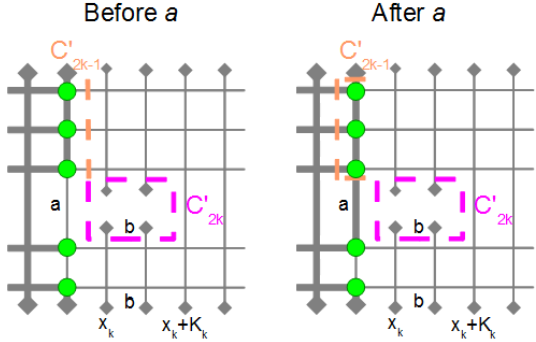}
\caption{(Color online) The part of a punctured cylinder graph around the $k^{th}$ hole.  Two stages are depicted, just before and just after the edge $a$ is measured.  The set $\widetilde{E}$ is shown in bold, and the vertices in $\partial \widetilde{E}$ are marked by circles(green). A choice of the non-trivial cocycles $C'_{2k-1}$ and $C'_{2k}$ that are convenient for each step are shown as dotted lines (orange and purple, respectively).}
\label{fig:holeclasses}
\end{figure}

\begin{proof} On the left side of Figure \ref{fig:holeclasses}, we show the relevant encoding cocycles $C'_{2k-1}$ and $C'_{2k}$, chosen in accordance with 
Figure \ref{homcircuitgraphfig}. From this and the \textbf{LtoR} ordering, it is clear that as soon as the edge $a$ is measured, the set $B$ is no longer empty. Rather, $B=\{2k-1\}$ i.e., there exists no $\hat{x}\in E_0(\hat{E})$ such that $|\hat{x} \cap C'_{2k-1}|=1$, yet $C'_{2k-1} \cap \hat{E} \ne \emptyset$. This is because there is no cycle that can ``wrap around'' the $k^{th}$ hole without using the edge $a$ or one to its left. With $B\ne \emptyset$, Equation \ref{prob} becomes more complicated. However, we can avoid this by considering the alternative encoding cocycle $C'_{2k-1}$ depicted on the right side of Figure \ref{fig:holeclasses} as soon as the edge $a$ is measured. This cocycle is homologous to the first (they differ only by the bitwise addition of $\delta_v$ for a set of vertices $v$) and hence their effect on the surface-code space is identical.

With $C'_{2k-1}$ chosen in this way, we have $A=\{2k...2g\}$ and $B=\emptyset$. Furthermore, $C'_j \in \widetilde{E}$ for all $j=1...2k-1$. Equation \ref{prob} takes the form, like Equation \ref{probbtwholes}
\begin{eqnarray}
p\left(|\phi_{\widetilde{E}}\rangle\right) &=& {1\over{\sqrt{2^{|\partial\widetilde{E}|-1}}}}  \sum_{\gamma, \delta} c_{\gamma}c^*_{\delta}\prod_{j=2k}^{2g}\delta_{\gamma_j,\delta_j}\nonumber\\
&&\langle\phi_{\widetilde{E}}\otimes \phi^*_{\widetilde{E}}|\bar{Z}^{\gamma}_{\widetilde{E}_1}\bar{Z}^{\delta}_{\widetilde{E}_2}|K(G(\widetilde{E}_1)\cup G(\widetilde{E}_2))\rangle\nonumber.\\ \label{probcrossingholes}
\end{eqnarray}

What remains now is to define a natural embedding of the graph $G(\widetilde{E}_1)\cap G(\widetilde{E}_2)$, which requires a more complicated topology than in the case of measurements between holes. To aid in this, we will employ two graph manipulations that only affect the overlap between $|K(G)\rangle$ and a product state up to a constant of proportionality. For any connected graph $G$, we may perform the following operations: 
\begin{itemize}
 \item \textit{Edge addition}: We may add an edge $e$ to G, then measure the qubit associated with the added edge to be in the $|0\rangle$ state. The edge can be added between existing vertices on G, or by adding a new vertex and connecting it to G with the new edge. Call the new graph obtained after edge addition $G'$. Then: $\langle 0_e|K(G')\rangle={1\over\sqrt{2}}|K(G)\rangle$.
 \item \textit{Vertex splitting}: We can split any vertex into two, and add an edge $e$ in between the two resultant vertices. The edges incident on the vertex that is split can be divided arbitrarily between the two resultant vertices. Then measure the new qubit to be in the $|+\rangle$ state. Call the new graph obtained by vertex splitting $G'$. Then: $\langle +_e|K(G')\rangle={1\over{\sqrt{2}}}|K(G)\rangle$
\end{itemize}

Using edge addition and vertex splitting \footnote{We note that the edge addition and vertex splitting are exactly the opposite of the \textit{graph minor} operations: edge contraction, edge deletion, and deletion of isolated vertices.  This implies that the class of graphs under consideration is \textit{minor closed}. In principle, this means that our considerations on the punctured cylinder graphs are applicable to any graph, because the family of all punctured cylinder graphs contains every graph as a minor.  This follows from a result by Robertson and Seymour \cite{robertsonseymour} to the effect that for any two graphs $G$ and $H$ that can be embedded on a surface $S$ of genus $g\ge 1$, $H$ is a minor of $G$ if the face-width of $G$ is at least $k(H)$, where $k(H)$ is an integer that depends on the graph $H$.  Face-width is the minimum number of edges of a graph that any non-contractible loop on $S$ must cross, which is a controllable parameter within the family of punctured cylinder graphs. However, this result is not of practical use here without knowledge of how $k(H)$ scales with the size and genus of $H$.}, we transform the graph $G(\widetilde{E}_1)\cap G(\widetilde{E}_2)$ into an effective graph $G'(\widetilde{E})$ that has a natural embedding on a surface of genus $2k-1$. An example of this is shown in Figure \ref{fig:astep}.
\begin{figure}[h]
\includegraphics[width=2.5in]{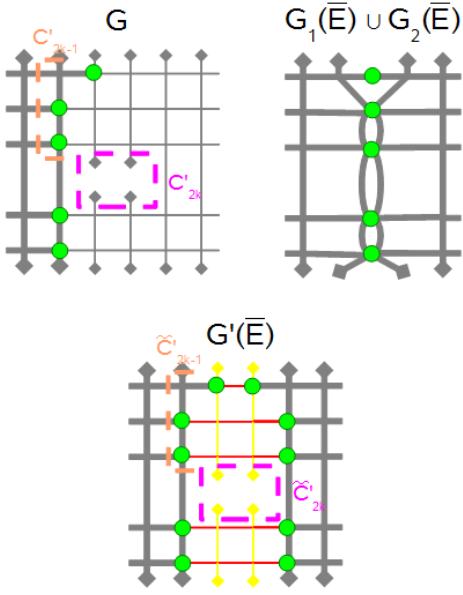}
\caption{(Color online) \label{fig:astep}The graphs $G$, $G(\widetilde{E}_1)\cap G(\widetilde{E}_2)$, and $G'(\widetilde{E})$ for a ``crossing hole" step of MBQC. The vertical nonbold edges of $G'(\widetilde{E})$ (yellow) are measured in the $|0\rangle$ state, while the horizontal nonbold edges(red) are measured in the $|+\rangle$ state. Encoding cocycles are shown for $G$ and $G'(\widetilde{E})$.}
\end{figure}

A surface code on $G'(\widetilde{E})$ encodes $4k-2$ qubits.  The encoding cocycles $\bar{C}'_1$...$\bar{C}'_{4k-2}$ on the embedded graph $G'(\widetilde{E})$ can be chosen as follows: let the first $2k-2$ cocycles be $\bar{C}'_j:=C'_j$ applied to the edges $\widetilde{E}_1$, and the last $2k-2$ cocycles be $\bar{C}'_{j+2k}:=C'_j$ applied to the edges $\widetilde{E}_2$.  The cocycle for qubit numbered $2k-1$ can be chosen as the cocycle $\bar{C}'_{2k-1}:=C'_{2k-1}$ applied to the edge set $\widetilde{E}_1$. Finally the cocycle $\bar{C}'_{2k}$ for qubit $2k$ belongs to the newly added edges, as depicted in Figure \ref{fig:astep}.

Let $\bar{E}$ denote the edges which are added to $G_1(\widetilde{E}) \cup G_2(\widetilde{E})$ to construct $G'(\widetilde{E})$, and let $|\bar{\phi}_{\bar{E}}\rangle$ denote a tensor product of the $|+\rangle$ state for each of the horizontal edges (added by vertex splitting), and $|0\rangle$ for each of the vertical edges (added by edge addition). We can recast Equation \ref{probcrossingholes} as 
\begin{eqnarray}
p\left(|\phi_{\widetilde{E}}\rangle\right) &=& {1\over{\sqrt{2^{|\partial\widetilde{E}|-|\bar{E}|-1}}}}  \sum_{\gamma, \delta} c_{\gamma}c^*_{\delta}\prod_{j=2k}^{2g}\delta_{\gamma_j,\delta_j}\nonumber\\
&&\langle\phi_{\widetilde{E}}\otimes\bar{\phi}\otimes \phi^*_{\widetilde{E}}|\bar{Z}^{\gamma}_{\widetilde{E}_1}\bar{Z}^{\delta}_{\widetilde{E}_2}|K(G'(\widetilde{E})\rangle\nonumber,\\ \label{probcrossingholes2}
\end{eqnarray}

where we can take $\bar{Z}^{\gamma}_{\widetilde{E}_1}\bar{Z}^{\delta}_{\widetilde{E}_2}$ to be
$$\prod_{e\in\bar{C}'_{2k-1}}Z_e^{(\gamma\oplus\delta)_{2k-1}}\prod_{j=1}^{2k-2}\left(\prod_{e\in\bar{C}'_{j}}Z_e^{\gamma_j}\prod_{e\in\bar{C}'_{j+2k}}Z_e^{\delta_j}\right)$$
which depends only on the bitwise sum $(\gamma\oplus\delta)_{2k-1}$ because the cocycle $C'_{2k-1}$ applied to the edge set $\widetilde{E}_2$ is homologous to $\bar{C}'_{2k-1}$ on the graph $G'(\widetilde{E})$. So, if both $\gamma_{2k-1}$ and $\delta_{2k-1}$ are equal to one, there is no overall effect on the state $|K(G'(\widetilde{E})\rangle$.

Now, since all of the edges in the set $\bar{C}'_{2k}$ are measured in the state $|0\rangle$, we may insert the operator $\bar{Z}_{2k}:=\prod_{e\in\bar{C}'_{2k}}Z_e^{\gamma_{2k}}$ with impunity. Then
$$\bar{Z}^{\gamma}_{\widetilde{E}_1}\bar{Z}^{\delta}_{\widetilde{E}_2}\prod_{e\in\bar{C}'_{2k}}Z_e^{\gamma_{2k}}|K(G'(\widetilde{E})\rangle$$
is precisely the encoded X eigenstate $$|X^{G'(\widetilde{E})}_{\gamma_1...\gamma_{2k-2},(\gamma\oplus\delta)_{2k-1},\gamma_{2k},\delta_1...\delta_{2k-2}}\rangle$$ in the surface-code space of $G'(\widetilde{E})$. If we now define
\begin{eqnarray}
\bar{c}_{\gamma_1...\gamma_{2k},\delta_1...\delta_{2k-2}} &:=& \sum_{\substack{\gamma_{2k+1}...\gamma_{2g}\\ \delta_{2k-1}...\delta_{2g}}} c_{\gamma_1...\gamma_{2k-2},(\gamma\oplus\delta)_{2k-1},\gamma_{2k}...\gamma_{2g}}\nonumber\\
&&c^*_{\delta_1...\delta_{2g}}\prod_{j=2k}^{2g}\delta_{\gamma_j,\delta_j} \label{cbardef2},
\end{eqnarray}
then the probability of a outcome on the edges in $\widetilde{E}$ from the original graph is exactly proportional to an inner product with a state in the code space of the surface code on $G'(\widetilde{E})$:
\begin{eqnarray}
p\left(|\phi_{\widetilde{E}}\rangle\right) &=& {1\over{\sqrt{2^{|\partial\widetilde{E}|-1}}}}\langle\phi_{\widetilde{E}}\otimes \phi^*_{\widetilde{E}}|\nonumber\\
&& \sum_{\substack{\gamma_{1}...\gamma_{2k}\\ \delta_{1}...\delta_{2k-2}}}\widetilde{c}_{\gamma_1...\gamma_{2k},\delta_1...\delta_{2k-2}}|X^{G'(\widetilde{E})}_{\gamma_1...\gamma_{2k},\delta_1...\delta_{2k-2}}\rangle\nonumber,\\ \label{probcrossingholesfinal}
\end{eqnarray}
which again takes the form of the inner product between a surface-code state and product state.  One can find a suitable $G'(\widetilde{E})$ to put $p\left(|\phi_{\widetilde{E}}\rangle\right)$ into the form of Equation \ref{probcrossingholesfinal} at all MBQC stages while crossing a hole; we have shown just one example of such a stage. During later stages the encoding cocycle $C'_{2k}$ will be split across the measured and unmeasured edges: $C'_{2k}\cap \widetilde{E}\ne\emptyset$ and $C'_{2k}\cap \hat{E}\ne\emptyset$.  However, we can always still ``complete'' the partial cocycle $C'_{2k}\cap \widetilde{E}$ from $G$ to a cocycle $\bar{C}'_{2k}$ on $G'(\widetilde{E})$ by adding edges from $\bar{E}$ that are measured in the $|0\rangle$ state. An example of this is shown in Figure \ref{fig:anotherstep}. Note that given our ordering of measurements, there still exists an $\hat{x} \in E_0(\hat{E})$ such that $|\hat{x}\cap C'_{2k}|=1$ until the edge $b$ from Figure \ref{fig:holeclasses} is measured.  Yet, once $b$ is measured $C'_{2k}\in \widetilde{E}$, so $2k \notin B$ and Equation \ref{probcrossingholesfinal} holds for all stages. This completes the proof of Theorem \ref{mbqctheorem} for all stages of computation.
\begin{figure}[t]
\includegraphics[width=2.5in]{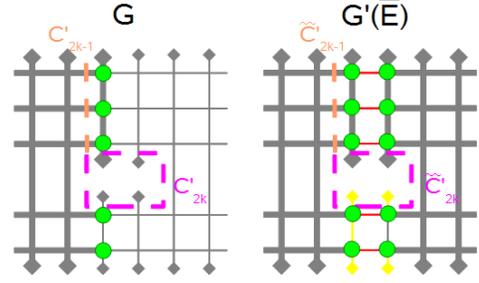}
\caption{(Color online) \label{fig:anotherstep}The graphs $G$ and $G'(\widetilde{E})$ during a stage of MBQC where $C'_{2k}$ is split across $\widetilde{E}$ and $\hat{E}$.}
\end{figure}
\end{proof}

\section{MBQC with the states $|\bar{C}^{\alpha,\beta}\rangle$} \label{specialstatesapp}

With Theorem \ref{mbqctheorem}, we have reduced the problem of simulating MBQC on punctured cylinder code states with \textbf{LtoR} to the evaluation of an inner product
\begin{equation} \label{effstate}
\langle\phi(G'(\widetilde{E}))|\left(\sum_{\gamma \in \{0,1\}^{\otimes 2g'}} \widetilde{c}_{\gamma}|X^{G'(\widetilde{E})}_\gamma\rangle\right),
\end{equation}
where $G'(\widetilde{E})$ is an effective lattice of genus $g'=2k$ or $2k-1$, $k$ is the number of holes in the set of qubits that have already been measured, and $|\phi(G'(\widetilde{E}))\rangle$ is a product state. Recall that in the associated encoded X-eigenbasis(corresponding to a canonical polygonal schema), the state $|\bar{C}^{\alpha,\beta}\rangle$ has coefficients
$$c_{\gamma, \rho} : = {1\over{2^g}}\prod_{j=1}^{g}(-1)^{\alpha_j\beta_j+(\alpha\oplus\gamma)_j(\beta\oplus\rho)_j},$$
where the notation $c_{\gamma, \rho}$ separates the odd and even numbered encoded qubits into two g-component bitstrings $\gamma$ and $\rho$. Here we will show that for MBQC with punctured cylinder code states $|\bar{C}^{\alpha,\beta}\rangle$, the tensor $\widetilde{c}_{\gamma, \rho}$ takes this same form, and thus the state 
$$\sum_{\gamma, \rho \in \{0,1\}^{\otimes g'}} \widetilde{c}_{\gamma, \rho} |X^{G'(\widetilde{E})}_{\gamma,\rho}\rangle$$
in Equation \ref{effstate} can be interpreted as a state $|\bar{C}^{\alpha',\beta'}\rangle$ in the code space of the surface code on the effective graph $G'(\widetilde{E})$, for some $\alpha', \beta' \in \{0,1\}^{\otimes g'}$. Then the efficiency of sampling follows by Equation \ref{specialstateeq}. Here the notation associates $\gamma$ with the even numbered qubits and $\rho$ with the odd: e.g. $\widetilde{c}_{\gamma, \rho}:=\widetilde{c}_{\gamma_1,\rho_1,\gamma_2..\rho_{g'}}$ (note the possible confusion with Equations \ref{probbtwholesfinal} and \ref{probcrossingholes}). 

To verify the above claim, we begin with the case where computation is between holes.  Using the definition of the $\widetilde{c}$ coefficients (Equation \ref{cbardef1}), after the summation $\widetilde{c}_{\gamma_1...\gamma_{k}\delta_1...\delta_{k},\rho_1...\rho_{k}\epsilon_1...\epsilon_k}$ works out to be:
\begin{eqnarray*}
{1\over{2^{2k}}}\prod_{j=1}^{k}(-1)^{\alpha_j\beta_j+(\alpha\oplus\gamma)_j(\beta\oplus\rho)_j}(-1)^{\alpha_j\beta_j+(\alpha\oplus\delta)_j(\beta\oplus\epsilon)_j}.
\end{eqnarray*}
This is exactly the tensor of coefficients for the state $|\bar{C}^{\alpha',\beta'}\rangle$ in the code space of a punctured cylinder code with $2k$ slots, labelled by bitstrings that are symmetric between the first and last $k$ entries: $\alpha':=\alpha\&\alpha, \beta':=\beta\&\beta$, where $\&$ denotes concatenation. The encoded Z cocycles are again those of a canonical encoding scheme, so local overlaps with $|\bar{C}^{\alpha',\beta'}\rangle$ can be computed efficiently in $|E|$ and $g$.

When crossing holes, we perform the summation of Equation \ref{cbardef2} for $\widetilde{c}_{\gamma_1...\gamma_{k}\delta_1...\delta_{k-1},\rho_1...\rho_{k}\epsilon_1...\epsilon_{k-1}}$  to obtain:
\begin{eqnarray*}
&&{1\over{2^{2k}}}\prod_{j=1}^{k-1}(-1)^{(\alpha\oplus\gamma)_j(\beta\oplus\rho)_j}(-1)^{(\alpha\oplus\delta)_j(\beta\oplus\epsilon)_j}\\
&&\sum_{\delta_k\in\{0,1\}}(-1)^{(\alpha\oplus\gamma\oplus\delta)_k(\beta\oplus\rho)_k}(-1)^{(\alpha\oplus\delta)_k(\beta\oplus\rho)_k}\\
&=&{1\over{2^{2k-1}}}\prod_{j=1}^{k-1}(-1)^{(\alpha\oplus\gamma)_j(\beta\oplus\rho)_j}(-1)^{(\alpha\oplus\delta)_j(\beta\oplus\epsilon)_j}\\
&&(-1)^{\gamma_k(\beta\oplus\rho)_k},
\end{eqnarray*} 
which is again the tensor describing $|\bar{C}^{\alpha',\beta'}\rangle$ in the code space of the surface code for $G'(\widetilde{E})$, where $\alpha':=\alpha_1,...\alpha_{k-1},0,\alpha_1,...\alpha_{k-1}$ and $\beta':=\beta_1,...\beta_k,\beta_1...\beta_{k-1}$. $\Box$

%\begin{acknowledgments}
%\end{acknowledgments}

%\bibliography{biblio}

\bibliography{biblio}

\end{document}